\documentclass[notitlepage,aps,pra,superscriptaddress,nofootinbib,10pt,showpacs]{revtex4-1}
\newif\ifarxiv\arxivtrue
\arxivfalse

\usepackage[margin=2.5cm]{geometry}

\usepackage{hyperref}

\usepackage{amsmath}
\usepackage{amssymb}
\usepackage{graphicx}
\usepackage{hyperref}
\usepackage{color}

\usepackage{xcolor,soul}

\definecolor{ForestGreen}{RGB}{34,139,34}

\hypersetup{
	colorlinks,
	linkcolor={black!30!blue},
	citecolor={red},
	urlcolor={blue}
}

\usepackage[T1]{fontenc}
\usepackage[latin9]{inputenc}

\usepackage{amsthm}
\usepackage{bbm}
\usepackage{mathtools}
\mathtoolsset{centercolon}
\ifarxiv\fi

\usepackage{xspace}

\usepackage{pgfplots}
\usepgfplotslibrary{patchplots}
\usepackage{tikz}
\usetikzlibrary{arrows}
\usetikzlibrary{decorations.pathreplacing}
\usetikzlibrary{decorations.markings}
\usetikzlibrary{shapes,arrows,positioning,patterns}
\usetikzlibrary{matrix}
\usetikzlibrary{fit}

\def\eat#1{}

\newtheorem{thm}{Theorem}[section]
\newtheorem{defi}[thm]{Definition}
\newtheorem{lem}[thm]{Lemma}
\newtheorem{prop}[thm]{Proposition}
\newtheorem{cor}[thm]{Corollary}

\def\splitref#1:#2@{\thing#1@\;\ref{#1:#2}}
\def\thing#1#2@{\ifx#1sSect.\else\ifx#1fFig.\else\ifx#1TTable\else\ifx#1tThm.\else\ifx#1dDef.\else\ifx#1lLem.\else\ifx#1cCor.\else\ifx#1pProp.\fi\fi\fi\fi\fi\fi\fi\fi}

\def\tr{\operatorname{tr}}

\def\Cx{{\mathbb C}}
\def\Ir{{\mathbb Z}}
\def\Tw{{\mathbb T}}

\def\mod{{\mathop{\rm mod}\nolimits}}

\def\tr{\mathop{\rm tr}\nolimits}

\def\rank{{\mathop{\rm rank}\nolimits}\,}

\let\tw\widetilde

\def\re{\Re e}

\def\dff{\bf} 

{{\def\HH{{\mathcal H}}\def\KK{{\mathcal K}}

\makeatletter
\newcommand{\raisemath}[1]{\mathpalette{\raisem@th{#1}}}
\newcommand{\raisem@th}[3]{\raisebox{#1}{$#2#3$}}
\makeatother

\def\six{\mathop{\mathrm s\mkern-1mu\mathrm i}\nolimits}
\DeclareRobustCommand{\sixnp}{\mathop{\mathrm s\mkern-1mu \textup \i}\nolimits}
\DeclareRobustCommand{\sixR}{\mathop{\hbox{\raisebox{.45em}{$\scriptstyle\rightharpoonup$}}\hspace{-.9em}\sixnp}}
\DeclareRobustCommand\sixL{\mathop{\hbox{\raisebox{.45em}{$\scriptstyle\leftharpoonup$}}\hspace{-.9em}\sixnp}}

\def\sixrel#1:#2{\six(#1{:}#2)}
\def\romkern{\kern-1pt}
\def\rom#1{\textnormal{I\if#11\else\romkern I\if#12\else\romkern I\fi\fi}}

\def\symS{\textnormal{\textsf S}\xspace}

\ifarxiv\else\fi

\def\ph{\eta}\def\rv{\tau}\def\ch{\gamma}

\newcommand{\supp}{\operatorname{supp}}

\usepackage{mathrsfs}

\def\1{\mathbbm{1}}

\def\Dk{{\mathbb D}}

\def\TT{{\mathcal T}}
\def\SS{{\mathcal S}}

\def\bs{\boldsymbol}

\def\>{\rangle}
\def\<{\langle}
\def\range{\operatorname{range}}
\def\spn{\operatorname{span}}
\def\sgn{\operatorname{sgn}}

\begin{document}

\title{Quantum walks: Schur functions meet symmetry protected topological phases
}
\nopagebreak

\author{C. Cedzich}
\affiliation{Institut f\"ur Theoretische Physik, Leibniz Universit\"at Hannover, Appelstr. 2, 30167 Hannover, Germany}
\affiliation{Institut f\"ur Theoretische Physik, Universit\"at zu K\"oln, Z\"ulpicher Stra{\ss}e 77, 50937 K\"oln, Germany}
\author{T. Geib}
\affiliation{Institut f\"ur Theoretische Physik, Leibniz Universit\"at Hannover, Appelstr. 2, 30167 Hannover, Germany}
\author{F.~A. Gr\"unbaum}
\affiliation{Department of Mathematics, University of California, Berkeley CA 94720}
\author{L. Vel\'azquez}
\affiliation{Departamento de Matem\'{a}tica Aplicada \& IUMA,  Universidad de Zaragoza,  Mar\'{\i}a de Luna 3, 50018 Zaragoza, Spain}
\author{A.~H. Werner}
\affiliation{{QMATH}, Department of Mathematical Sciences, University of Copenhagen, Universitetsparken 5, 2100 Copenhagen, Denmark,}
\affiliation{{NBIA}, Niels Bohr Institute, University of Copenhagen, Denmark}
\author{R.~F. Werner}
\affiliation{Institut f\"ur Theoretische Physik, Leibniz Universit\"at Hannover, Appelstr. 2, 30167 Hannover, Germany}

\begin{abstract}
This paper uncovers and exploits a link between a central object in harmonic analysis, the so-called Schur functions, and the very hot topic of symmetry protected topological phases of quantum matter. This connection is found in the setting of quantum walks, i.e. quantum analogs of classical random walks. We prove that topological indices classifying symmetry protected topological phases of quantum walks are encoded by matrix Schur functions built out of the walk. This main result of the paper reduces the calculation of these topological indices to a linear algebra problem: calculating symmetry indices of finite-dimensional unitaries obtained by evaluating such matrix Schur functions at the symmetry protected points $\pm1$. The Schur representation fully covers the complete set of symmetry indices for 1D quantum walks with a group of symmetries realizing any of the symmetry types of the tenfold way. The main advantage of the Schur approach is its validity in the absence of translation invariance, which allows us to go beyond standard Fourier methods, leading to the complete classification of non-translation invariant phases for typical examples.
\end{abstract}

\pacs{03.65.Vf,  
03.65.Db;  
\hskip5pt
MSC-class: 47A56, 
81Q99, 
30J99 
}

\maketitle

\section{Introduction}
\label{sec:int}

Topological phases of matter are currently one of the most stimulating topics in quantum physics \cite{kitaevPeriodic,KitaevLectureNotes,MPSphaseI,MPSphaseII,Schnyder1,Schnyder2,HasanKaneReview,KaneMeleQSH,KaneMeleTopOrder,ZhangTopologicalReview,Feffer}, leading to the Nobel prize in 2016. Theoretically predicted decades earlier, their experimental realization had to wait until 2007 \cite{konig2007quantum}. Since then, the impact of the potential applications, ranging from spintronics to topological superconductivity or quantum computation \cite{KitaevLectureNotes,KitaevModel,KitaevAnyons,AnyonsUniversalQuantumComp}, has fueled the race towards the discovery of new forms of topological matter, fostering a spectacular symbiosis between the theoretical and experimental efforts.
The key feature of these new states of matter is their exceptional stability against perturbations, which is of topological origin. The archetype here is the so called bulk-edge principle: the interface between systems in different phases supports robust bound states, regardless of how the phases are joined.

The precise meaning of such statements depends crucially on the specific framework. In this paper the physical systems are discrete-time evolutions of single particles with internal degrees of freedom, so-called ``quantum walks'' \cite{Aharonov,Grimmet,Ambainis2001,electric,SpaceTimeCoinFlux,Anonymous:rSnqW2vc,gauge}. They are described by a unitary operator $W$ on a lattice satisfying a locality condition which reflects the decay of the interactions as a function of the distance. Their simplicity and versatility make them ideal platforms for creating and studying new symmetry protected topological phases \cite{Kita,Kita2,short,long,ti,Asbo1,Asbo2,Asbo4}.

In one spatial dimension, the mathematical description of these phases involves a set $\mathfrak U$ of unitaries with some common symmetries --giving a representation of any one of the symmetry types of the tenfold way \cite{Altland-Zirnbauer}-- and satisfying some locality and spectral gap conditions \cite{short,long}. These constraints divide $\mathfrak U$ into homotopy classes (topological phases), so that a perturbation $t\mapsto W(t)\in\mathfrak U$ cannot change the homotopy class unless a discontinuity (phase transition) arises. Stated in another way, if a symmetry preserving continuous path $W(t)$ of walks --i.e. unitaries satisfying the locality conditions-- connects walks $W_1$, $W_2$ in different homotopy classes, then $W(t)$ closes the spectral gap for some $t$. The bulk-edge principle sketched above becomes a stability result on eigenvalues and eigenvectors of crossovers $W$ acting as $W_1\in\mathfrak U$ or $W_2\in\mathfrak U$ on different regions of the lattice: if $W_1,W_2$ are in different homotopy classes, the walk $W$ exhibits eigenvalues in the gaps with eigenvectors localized around the interface between $W_1$ and $W_2$ which cannot be eliminated by perturbing the system continuously.

A central problem in this context is to find a complete set of labels for these topological phases. In the purely translation invariant case  Fourier analysis is widely used to represent an infinite-dimensional quantum walk as a finite-dimensional matrix valued function providing simple expressions for such homotopy invariants \cite{ti}. A more sophisticated machinery which applies also in the disordered case is provided by K-theory \cite{kitaevPeriodic,Thiang,SchulzBaldesBook,SchulzZ2,Schulz2016index,SchulzNonTechnicalOverview}, but its complexity makes it difficult to apply this in practice.

One of the novel ideas in this work is the proposal of new tools for the analysis of symmetry protected topological indices in quantum walks which combine both, the simplicity of Fourier and the universality of K-theory. Known by the name of Schur functions, they constitute one of the gems resulting from the interplay between harmonic analysis and complex variables,
which go back to the early XXth century in the hands of Issai Schur \cite{Schur1,Schur2} (quick introductions to scalar and matrix valued Schur functions can be found in \cite[Section 1.3]{Simon} and \cite[Chapters 1 and 3]{DaPuSi}, for a more general and detailed treatment see for instance \cite{Dubovoj}). The flexibility of Schur functions is evidenced by a wide variety of applications such as interpolation problems, orthogonal polynomials, operator theory, system and control theory, stochastic processes, electrical engineering, signal processing or geophysics (see \cite{Simon,DaPuSi,NagyFoias,FoiasFrazho,Kailath} and references therein).

Regarding their quantum significance, Schur functions play a crucial role in describing recurrence properties of quantum walks --more generally, discrete-time unitary evolutions \cite{GVWW,BGVW}--, in the spirit of P\'olya's recurrence theory for classical random walks \cite{Polya,Feller}. The main message of this paper is that, once again, Schur functions are a very useful tool as they provide a bridge between harmonic analysis and the study of symmetry protected topological phases of quantum walks. This Schur approach to symmetry protected topological phases is based on a recent theory \cite{short,long} which avoids any translation invariance assumption and gives a complete set of topological indices for one-dimensional quantum walks in $\mathfrak U$.

As we shall see, matrix valued Schur functions do a similar job for non-translation invariant quantum walks as Fourier transform does in the periodic case, reducing the calculation of topological indices to finite-dimensional linear algebra problems. In a sense, one can even think of Schur functions as nonlinear versions of Fourier tools. This is illustrated by their relation to orthogonal polynomials on the unit circle \cite{Simon,DaPuSi}, which generalize standard Fourier expansions beyond the Lebesgue measure on the unit circle. Moreover, the continuous version of Schur functions, which plays a prominent role in Dirac's and Krein's systems \cite{Denisov}, appears in connection with appropriate ``scattering transforms'', regarded as nonlinear versions of the Fourier transform, a situation similar to that in the study of many integrable systems.

Therefore, the relevance of the Schur approach relies on the capability to go beyond translation invariance, which is important for several reasons:
\begin{itemize}
\item It is required by a rigorous treatment of the bulk-edge principle \cite{short,long} because joining different phases breaks translation invariance.
\item The bulk-edge principle guarantees the emergence of symmetry protected eigenvectors even when joining non-translation invariant bulks (see the end of Sect.~\ref{sec:ssrep}).
\item Although the general theory states that every phase includes a crossover $W$ of translation invariant walks $W_1$ and $W_2$ \cite{long}, changing $W_i$ by a translation invariant walk in the same phase may change the phase of the crossover. Hence, the knowledge of the translation invariant phases of a model is of little help to surmise its total number of phases (see the examples in Sect.~\ref{sec:ssrep} and \ref{sec:ex}).
\item Dealing with non-translation invariant walks unifies the treatment of all translation invariant cases, i.e. periodic situations with arbitrary period, something which seems hard to tackle with Fourier methods.
\end{itemize}

Finally, we have to remark that the Schur bridge between harmonic analysis and quantum walks goes both ways. This has been already the case for quantum recurrence which has spurred new results on Schur and Nevanlinna functions, as well as in related areas such as orthogonal polynomials \cite{CGVWW,GV}. Regarding possible payoffs of the Schur approach to topological phases, it could open new ways of thinking about homotopy groups of operator spaces, the general stability problem of isolated eigenvalues, or extensions of the Gohberg-Krein formula for Fredholm indices, closely related to the topological indices of quantum walks \cite{long,ti}.

The paper is structured as follows: Sect.~\ref{sec:SPTP} summarizes the general theory of symmetry protected topological phases for 1D quantum walks given in \cite{short,long}. A selection of results on Schur functions of interest for our purpose is given in Sect.~\ref{sec:FR-S}. This section contains a new result for matrix valued Schur functions, Theorem~\ref{thm:m-eg}, which is the cornerstone of the Schur approach to topological indices developed in Sect.~\ref{sec:S-I} and \ref{sec:dec}, see Theorems~\ref{thm:si}, \ref{thm:eigC}, \ref{thm:dec} and Corollary~\ref{cor:dec}. Sect.~\ref{sec:ssrep} and \ref{sec:ex} apply the Schur machinery to the complete classification of symmetry protected topological phases in examples of non-translation invariant quantum walks: the split-step walk, coined walks with coins of arbitrary size and some unitary equivalent transformations thereof.

\section{Symmetry protected topological phases of 1D quantum walks}
\label{sec:SPTP}

Let us summarize the main results on the topological classification of 1D symmetric quantum walks, first announced in \cite{short}, and fully developed in \cite{long}. Unlike other approaches, this provides an explicit set of indices which completely classifies the phases without any translation invariance assumption --concerning the special features of the translation invariant case see \cite{ti}--, getting such indices close to the Schur representation that we will develop.


We will deal with unitary operators $W$ on the line, i.e. on a Hilbert space $(\HH,\<\;\,|\;\,\>)$ with cell structure
$$
 \HH = \bigoplus_{x\in\Ir} \HH_x,
$$
where the {\bf cells} $\HH_x$ are finite-dimensional subspaces, but not necessarily with the same dimension since no assumption about translation invariance shall be made. The index $x\in\Ir$ labels the sites of a 1D lattice, while the cells $\HH_x$ describe internal degrees of freedom.
The identification of certain symmetry indices will require to deal also with unitaries on a left/right half-line, i.e. on a sum of cells to the left/right of a given one. Due to this reason, for convenience we will use the abbreviations
$$
 \HH_{>a} = \bigoplus_{x>a} \HH_x,
 \qquad
 \HH_{\ge a} = \bigoplus_{x\ge a} \HH_x,
 \qquad
 \HH_{<a} = \bigoplus_{x<a} \HH_x,
 \qquad
 \HH_{\le a} = \bigoplus_{x\le a} \HH_x.
$$
Besides, $P_x$ and $P_{>a}$ will denote the orthogonal projections of $\HH$ onto $\HH_x$ and $\HH_{>a}$, respectively, and analogously for the remaining cases.

To study properly symmetry protected topological phases in 1D quantum walks we need to consider unitaries $W$ on a (finite or infinite) sum of cells $\HH_x$ under the following setting:

\begin{itemize}

\item[(a)] The unitary $W$ has a discrete group of symmetries providing a concrete representation of one of the symmetry types of the {\it tenfold way} \cite{Altland-Zirnbauer}.
Every symmetry $\sigma$ acts locally in each cell, i.e. $\sigma=\oplus_x\sigma_x$ where $\sigma_x$ acts on $\HH_x$. Besides, the representation of the symmetry type in each cell is {\bf balanced}, i.e. every $\sigma_x$ has a symmetry invariant unitary with spectral gaps around the symmetry protected points $\pm1$. This requirement guarantees that what we will call left and right indices are well defined.

\item[(b)] The essential spectrum of $W$ has gaps around $\pm1$, i.e. the spectrum of $W$ around $\pm1$ consists only of isolated eigenvalues of finite multiplicity. This, instead of a strict gap, gives room for the presence of symmetry protected eigenvectors with eigenvalues $\pm1$. We will refer to a (strictly) {\bf gapped} or {\bf essentially gapped} unitary to distinguish both situations.

\item[(c)] A mild locality assumption will be made, namely, that $[W,P_{\ge a}]=WP_{\ge a}-P_{\ge a}W$ is compact, a property which is independent of $a$ and we call this {\bf essential locality}. This condition is equivalent to the compactness of $P_{<a}WP_{\ge a}$ and $P_{\ge a}WP_{<a}$, as follows from $[W,P_{\ge a}] = P_{<a}WP_{\ge a}-P_{\ge a}WP_{<a}$.

\end{itemize}

A unitary $W$ satisfying (a) and (b) will be called an {\bf admissible unitary} for the given representation of a symmetry type. Under the additional restriction (c), we will refer to such a unitary as an {\bf admissible walk}. When $P_yWP_x=0$ for $|y-x|$ greater than some $x$-independent length $L$, a walk $W$ will be called (strictly) {\bf local}. Nevertheless, assumption (c) leaves room for non-local walks on the line, e.g. all those with decay $\|P_yWP_x\|\le c|y-x|^{-\alpha}$ for some $c>0$ and $\alpha>1$ \cite{ti}.

A walk $W$ on the line is called {\bf translation invariant} if $[W,S_a]=0$ for some shift operator $S_a$, $a\in\Ir$, which refers to a unitary operator on $\HH$ such that $S_a\HH_x=\HH_{x+a}$. No translation invariant assumption is made in the characterization of symmetry protected topological phases described below, which holds for arbitrary admissible walks. For a translation invariant admissible walk $W$ the essential gap and strict gap conditions coincide because then no eigenspace of $W$ may be finite-dimensional. Therefore, any discussion of symmetry protected eigenvectors with eigenvalues $\pm1$ requires breaking translation invariance.

Every symmetry $\sigma$ of any of the alluded to symmetry types is represented, up to a phase, by an involution on the Hilbert state space $\HH$ acting on operators as $X \mapsto \sigma X \sigma^*$, where $\sigma^*$ stands for the adjoint of $\sigma$. These involutions are of one of the following three types:
\begin{center}
		\begin{tabular}{|c|c|c|}
			\dff particle-hole 	& \dff time-reversal 	& \dff chiral
			\\
			\hline
			antiunitary    		& antiunitary  			& unitary
			\\
			$\ph W\ph^*=W$    	& $\rv W\rv^*=W^*$  	& $\ch W\ch^*=W^*$
		\end{tabular}
\end{center}
The anitiunitary involutions necessarily satisfy\footnote{Here and in what follows $\1$ stands for the identity on the whole Hilbert space $\HH$, while we use the notation $\1_{\HH_C}$ for the identity on a subspace $\HH_C\subset\HH$.} $\sigma^2=\pm\1$ because $\sigma^2\sigma=\sigma\sigma^2$, while it is possible to adjust the phase convention for the symmetries so that $\ch^2=\1$ or $\ch=\ph\rv$ and the three involutions commute whenever all of them are present, in which case $\ch^2=\ph^2\rv^2$. Therefore, each {\bf symmetry type} $\symS$ of the {\it tenfold way} is determined by a commutative group of order 1, 2 or 4 constituted by unitary/antiunitary involutions, together with a choice for the signs of the squared antiunitary involutions whenever they are present. The particle-hole and chiral symmetries make the spectrum of a unitary $W$ invariant under complex conjugation, which distinguishes the $\pm1$-eigenspaces $\HH^\pm$ as the only ones invariant under any of the symmetries.

In this setting, the {\bf symmetry protected topological phases} are the homotopy classes of admissible walks under the norm topology. It has been proved in \cite{long} that these homotopy classes are labelled by three indices taking values in a group isomorphic to $\{0\}$, $\Ir_2$ or $\Ir$, given in terms of a {\bf symmetry index} $\six(\rho)$ classifying --up to unitary equivalence and orthogonal sum of balanced representations-- the representations $\rho$ of a symmetry type $\symS$ on a finite-dimensional Hilbert space $\HH$ \cite{short,long}. A group theoretical analysis yields \cite{short,long}
\begin{equation} \label{eq:SIfin}
 \six(\rho) =
 \left\{
 \begin{aligned}
 	& \dim\HH \; (\mod\,2)
	& \quad & \text{for } \symS = \{\1,\ph\},
	\; \ph^2=\1,
	& \qquad & \text{(I)}
	\\
	& \dim\HH \; (\mod\,4)
	& & \text{for } \symS = \{\1,\ph,\rv,\ch\},
	\; \ph^2=\1, \; \rv^2=-\1,
	& & \text{(II)}
	\\
	& \tr\ch
	& & \text{for } \symS = \{\1,\ch\} \text{ or } \{\1,\ph,\rv,\ch\},
	\; \ph^2=\rv^2 \; (\text{i.e.} \; \ch^2=\1),
	& & \text{(III)}
	\\
	& 0,
	& & \text{otherwise},
 \end{aligned}
 \right.
\end{equation}
showing that $\six(\rho\oplus\rho')=\six(\rho)+\six(\rho')$, $\six(V\rho V^{-1})=\six(\rho)$ and $|\six(\rho)|\le\dim\HH$.

Given an admissible unitary $W$ for a representation $\rho$ of a symmetry type,
consider the representations $\rho_\pm$ induced by $\rho$ on the $\pm1$-eigenspaces $\HH^\pm$ of $W$. Any continuous perturbation of $W$ by admissible unitaries changes $\rho_\pm$ by adding/subtracting balanced representations, hence the following indices --which make sense since $\dim\HH^\pm<\infty$ by the essential gap condition-- are homotopy invariants in the set of admissible unitaries,
$$
 \six(W) = \six_+(W) + \six_-(W),
 \qquad
 \six_\pm(W)=\six(\rho_\pm).
$$
Also, $|\six_\pm(W)|\le\dim\HH^\pm$ so that $\six(W)\ne0$ guarantees the existence of an eigenvector with eigenvalue 1 or $-1$. In the finite-dimensional situation $\six(W)=\six(\rho)$ and $\six_\pm(W)$ admit the following direct expressions \cite{long},
\begin{equation} \label{eq:SI+-fin}
 \begin{aligned}
 	& (-1)^{\six_\pm(W)} = \det(\mp W),
	& \qquad & \text{(I)}
	\\
	& \six_\pm(W) = \frac{1}{2} \tr\ch(\1\pm W).
	& & \text{(III)}
 \end{aligned}
\end{equation}

For an admissible walk $W$ on the line, another pair of indices arise, which make the cell structure to come into play. They are given by
\begin{equation} \label{eq:LR}
 \sixL(W) = \six(W_L), \qquad \sixR(W)=\six(W_R),
\end{equation}
with $W_{L/R}$ a walk on a left/right half-line $\HH_{<b}$/$\HH_{\ge b}$, such that $W_L\oplus W_R$ is admissible and ``coincides with $W$ far to the left/right'', i.e.
\begin{equation} \label{eq:CD}
 \lim_{a\to-\infty}\|P_{<a}(W-(W_L\oplus W_R))P_{<a}\| =
 \lim_{a\to\infty}\|P_{>a}(W-(W_L\oplus W_R))P_{>a}\| = 0.
\end{equation}
Due to essential locality, \eqref{eq:CD} is equivalent to the compactness of $W-(W_L \oplus W_R)$. The existence of such a compact decoupling is a non-trivial result of the theory \cite{long}. In the case of usual local walks the decoupling is typically performed by a local perturbation, i.e. acting non-trivially only on a finite number of cells.

While $\six_\pm(W)$ make sense for any admissible unitary $W$, the left and right indices, $\sixL(W)$ and $\sixR(W)$, are well defined only for admissible walks since essential locality guarantees their invariance under compact perturbations. Essential locality is also behind the identity \cite{short,long}
$$
 \six(W) = \sixL(W) + \sixR(W),
$$
which shows that $\sixL(W)$, $\sixR(W)$ and $\six_\pm(W)$ are not independent. Instead, any three of these indices are independent and can be used to label the homotopy classes of admissible walks. Omiting some of these three indices yields a less fine classification, either by weakening the homotopy equivalence relation, or by enlarging the set of unitaries which may be used along a homotopy deformation \cite{long}:
\begin{equation} \label{eq:table-equivrel}
		\begin{tabular}{|c|c|c|}
		    \hline
		    Indices 	
			& Unitaries	
			& Equivalence relation
			\\ \hline
 			& &
 			\\[-12pt]
 			\hline
			\;$\sixL(W)$, $\sixR(W)$, $\six_-(W)$\; 	
			& admissible walks
			& homotopies
			\\
			\hline
			$\sixL(W)$, $\sixR(W)$ 	
			& admissible walks
			& \;homotopies \& compact perturbations\;
			\\
			\hline
			$\six_+(W)$, $\six_-(W)$ 	
			& \;admissible unitaries\;
			& homotopies
			\\	
			\hline
		\end{tabular}
\end{equation}
The indices $\sixL$, $\sixR$ and $\six_\pm$ are also invariant under the map $W \mapsto W^*$. Besides, $\six_\pm$ are invariant under unitary equivalence of the walk $W \mapsto UWU^*$ and the symmetries $\sigma \mapsto U\sigma U^*$. However, this is not the case for $\sixL$ and $\sixR$, unless the unitary $U$ acts locally in each cell so that it preserves the cell structure.

The above classification of topological phases is not a mere intellectual challenge, but has measurable physical consequences which follow from the {\bf bulk-edge correspondence} \cite{short,long}: Let $W$ be a crossover between admissible walks $W_1$, $W_2$ on the line, i.e. an admissible walk which coincides with $W_{1/2}$ far to the left/right,
$$
 \lim_{a\to-\infty}\|P_{<a}(W-W_1)P_{<a}\| =
 \lim_{a\to\infty}\|P_{>a}(W-W_2)P_{>a}\| = 0.
$$
Assuming $\six_\pm(W_i)=0$ --a condition guaranteed by translation invariance, or more generally for strictly gapped walks, but satisfied in more general situations--, then $\sixL(W_i)=-\sixR(W_i)$, thus the single index $\sixR(W_i)$ identifies the topological phase to which $W_i$ belongs. In this case, $\sixL(W)=\sixL(W_1)$, $\sixR(W)=\sixR(W_2)$, $\six(W) = \sixR(W_2)-\sixR(W_1)$ and the $\pm1$-eigenspaces $\HH^\pm$ of $W$ satisfy
$$
 |\sixR(W_2)-\sixR(W_1)| \le \dim\HH^+ + \dim\HH^-.
$$
Hence, $W$ has an eigenvector with eigenvalue 1 or $-1$ whenever $\sixR(W_1)\ne\sixR(W_2)$. These results hold in the absence of any translation invariance assumption for $W_i$. We conclude that any crossover between gapped admissible walks from different topological phases breaks the gap by creating protected eigenvectors with eigenvalue 1 or $-1$, regardless of the way the crossover is engineered. These eigenvectors are typically localized around the crossover, a phenomenon that can be precisely proved when $W$ interpolates between translation invariant walks due to the exponential decay of its eigenvectors \cite{ti}.

\section{Schur functions}
\label{sec:FR-S}

Obtaining the symmetry indices of a walk $W$ involves two ingredients: the decoupling of $W$ into left/right walks $W_{L/R}$ and the calculation --according to \eqref{eq:SIfin}-- of the symmetry indices $\six(\rho_\pm)$ for the finite-dimensional representations $\rho_\pm$ induced on the $\pm1$-eigenspaces of $W$ and $W_{L/R}$. However, getting information about eigenspaces of operators in infinite dimension is not an easy task, especially for non-translation invariant walks, whose treatment is among the most significant virtues of the theory previously summarized.

Under translation invariance this problem may be circumvented by the Fourier transform, which represents a walk as a matrix valued function on the {\bf unit circle} $\Tw=\{z\in\Cx:|z|=1\}$. The main contribution of this paper is the discovery that certain matrix valued functions on the {\bf unit disk} $\Dk=\{z\in\Cx:|z|<1\}$, known as Schur functions \cite{Schur1,Schur2} (for a modern approach see \cite{Dubovoj}, or the summaries in \cite[Section 1.3]{Simon} and \cite[Chapters 1 and 3]{DaPuSi}), do a similar job when translation invariance is broken. We will summarize here some of the special properties of these functions which are of interest for us.

Scalar {\bf Schur functions} are the analytic maps of $\Dk$ into its closure $\overline{\Dk}=\Dk\cup\Tw$. Among their most remarkable features is their characterization by a finite or infinite sequence of complex numbers $\alpha_n$ --the {\bf Schur parameters}-- arising from the so called {\bf Schur algorithm}
\begin{equation} \label{eq:SA}
 f_0=f;
 \qquad\quad
 f_{n+1}(z) = T_{\alpha_n} f_n(z) :=
 \frac{1}{z} \frac{f_n(z)-\alpha_n}{1-\overline\alpha_n f_n(z)},
 \qquad
 \alpha_n = f_n(0),
 \qquad n\ge0,
\end{equation}
which generates iteratively new Schur functions $f_n$ --the {\bf Schur iterates}-- starting from a given one $f$. This yields an infinite sequence of iterates unless some $\alpha_n$ lies in $\Tw$, which stops the algorithm, establishing a one-to-one correspondence between Schur functions and elements of $\Dk^\infty \cup (\cup_{k\ge0}\Dk^k\times\Tw)$. If $f$ is a Schur function with Schur parameters $(\alpha_n)_{n\ge0}$, obviously its $k$-th iterate $f_k$ has Schur parameters $(\alpha_n)_{n\ge k}$.

The Schur algorithm itself gives the following relations between transformations of Schur functions and Schur parameters:
\begin{equation} \label{eq:rSA}
\begin{aligned}
 & (\alpha_n)_{n\ge0} \mapsto (\lambda\alpha_n)_{n\ge0}
 \;\Rightarrow \; f(z) \mapsto \lambda f(z),
 \qquad \lambda\in\Tw,
 \\
 & (\alpha_n)_{n\ge0} \mapsto (\alpha_0,0,\alpha_1,0,\dots)
 \;\Rightarrow \; f(z) \mapsto f(z^2),
 \\
 & (\alpha_n)_{n\ge0} \mapsto (0,\alpha_0,0,\alpha_1,\dots)
 \;\Rightarrow \; f(z) \mapsto zf(z^2).
\end{aligned}
\end{equation}
Therefore, a change in the sign in all the Schur parameters changes the sign in the Schur function, while even/odd Schur functions are characterized by sequences of Schur parameters which are null at even/odd places in the sequence.

A Schur algorithm also exists for matrix Schur functions, the contractive matrix valued analytic functions on $\Dk$. Another fruitful property of these functions is their one-to-one correspondence with matrix measures $\mu$ on $\Tw$ normalized by $\mu(\Tw)=\1$. The matrix Schur function $f$ related to any such a measure is given by
\begin{equation} \label{eq:f-F}
 F(z) = \int \frac{t+z}{t-z} \, d\mu(t),
 \qquad
 f(z) = z^{-1} (F(z)-\1)(F(z)+\1)^{-1},
\end{equation}
where $F$, known as the {\bf Carath\'eodory function} of $\mu$, is a matrix valued analytic function on $\Dk$ with positive real part $\re F(z) = (F(z)+F(z)^*)/2$. The radial limits
$$
 f(e^{i\theta}) = \lim_{r\uparrow1} f(re^{i\theta}),
 \qquad
 F(e^{i\theta}) = \lim_{r\uparrow1} F(re^{i\theta}),
$$
exist a.e. on $\Tw$, and the asymptotic behaviour of Schur and Carath\'eodory functions at the boundary $\Tw$ yields information about the measure $\mu$. For instance, the absolutely continuous part of $\mu$ is given by $\re F(e^{i\theta})\,\frac{d\theta}{2\pi}$, and its singular part is concentrated on the points $e^{i\theta}$ such that $\lim_{r\uparrow1}\tr\re F(re^{i\theta})=\infty$.

Concerning the boundary behaviour, Schur functions have better analyticity properties than Carath\'eodory functions. The later ones are analytic on the gaps of the support of the measure, but have simple poles at the isolated mass points. In contrast, Schur functions have an analytic continuation through the gaps of the limit points of the support, including the isolated mass points. For scalar Schur functions, this follows from the fact that a quotient of meromorphic functions with the same poles, and of the same order, is necessarily analytic. This pole cancellation does not generalize to matrix valued functions, thus the analyticity of Schur functions at the isolated mass points is less trivial in the matrix valued case. A proof for matrix Schur functions is given in the following theorem which, for a matrix measure on $\Tw$, provides a Schur characterization for the isolated mass points and the gaps of the limit points of the support. Some partial results in the theorem below follow from known relations between matrix measures and matrix Carath\'eodory functions which are straight forward extensions of similar results for the scalar case \cite{GeTs}, and whose succint proof is included for completeness. Nevertheless, the final characterization in terms of matrix Schur functions is new and the proof is deliberately more explicit concerning those aspects which need a much more delicate treatment than in the scalar case. The relevance of this characterization for our purposes lies in the connection between eigenvalues/essential gaps of unitary operators and mass points/gaps of the limit points of the support for the related spectral measures (see Sect.~\ref{sec:S-I}).

\begin{thm} \label{thm:m-eg}
Let $\mu$ be a matrix measure with support $\supp\mu\subset\Tw$ and $\mu(\Tw)=\1$, and denote by $(\supp\mu)'$ the set of limit points of $\supp\mu$. If $f$ is the matrix Schur function of $\mu$, then:
\begin{itemize}
 \item[(i)] $\Tw\setminus(\supp\mu)'$ is constituted by the arcs of $\Tw$
    through which $f$ has an analytic continuation which is unitary on these
    arcs.
 \item[(ii)] The isolated mass points $\lambda$ of $\mu$ are the zeros of
    $\det(\1-zf(z))$ lying on $\Tw\setminus(\supp\mu)'$, and
	\begin{equation} \label{eq:muschur}
	 \range \mu(\{\lambda\}) = \ker(\1 - \lambda f(\lambda)).
	\end{equation}
\end{itemize}
\end{thm}

\begin{proof}
\noindent{\it(i)}
Let us prove first that $f$ extends analytically through $\Tw\setminus(\supp\mu)'$ and takes unitary values there.

The Carath\'eodory function $F$ of $\mu$ is obviously analytic on $\Cx\setminus\supp\mu$. Moreover, given an isolated mass point $\lambda$ of $\mu$, for some $\epsilon>0$ we have the splitting
\begin{equation} \label{eq:F01}
 F(z) =
 \frac{\lambda+z}{\lambda-z} \mu(\{\lambda\}) +
 \int_{|t-\lambda|>\epsilon} \frac{t+z}{t-z}\,d\mu(t) =
 F_0(z)+F_1(z),
\end{equation}
with $F_1$ analytic, not only on $\Cx\setminus\supp\mu$, but also at $\lambda$. Therefore, $F$ is a meromorphic function on $\Cx\setminus(\supp\mu)'$ whose singularities are simple poles located at every isolated mass point $\lambda$ with residue $-2\lambda\mu(\{\lambda\})$.

Let us see that $f$ has an analytic continuation through $\Tw\setminus(\supp\mu)'$. The positivity of $\re F$ guarantees the invertibility of $F+\1$ on $\Dk$ and, by analyticity, also in a neighbourhood of $\Tw\setminus\supp\mu$, implying the analyticity of $f$ in such a neighbourhood. It only remains to see that $f$ is also analytic in a neighbourhood of each isolated mass point $\lambda$ of $\mu$. For this purpose, consider the orthogonal projection $P$ onto $\range\mu(\{\lambda\})$ and its complementary $P^\bot=\1-P$, which is the orthogonal projection onto $\range\mu(\{\lambda\})^\bot=\ker\mu(\{\lambda\})$ because $\mu(\{\lambda\})$ is self-adjoint. They lead to the following block decomposition of the Carath\'eodory function,
$$
 F(z) =
 \begin{pmatrix}
 	PF(z)P & PF(z)P^\bot
	\\
	P^\bot F(z)P & P^\bot F(z)P^\bot
 \end{pmatrix}
 =
 \begin{pmatrix}
 	PF(z)P & PF_1(z)P^\bot
	\\
	P^\bot F_1(z)P & P^\bot F_1(z)P^\bot
 \end{pmatrix}
 =
 \begin{pmatrix}
 	\frac{1}{z-\lambda}A(z) & B(z)
	\\
	C(z) & D(z)
 \end{pmatrix},
$$
with $A,B,C,D$ analytic at $\lambda$. The block $(z-\lambda)^{-1}A(z)$ has a simple pole at $\lambda$ with residue $A(\lambda) = -2\lambda P\mu(\{\lambda\})P$, which is an automorphism of $\range\mu(\{\lambda\})$. Besides, $D$ is the Carath\'eodory function of the measure $P^\bot\mu P^\bot$ --the projection of $\mu$ onto $\ker\mu(\{\lambda\})$--, whose support lies in $\supp\mu\setminus\{\lambda\}$. Therefore, $D$ is analytic on $\Cx\setminus\supp\mu\cup\{\lambda\}$ and $D+\1_{\ker\mu(\{\lambda\})}$ is invertible in a neighbourhood of $\Tw\setminus\supp\mu\cup\{\lambda\}$, so that $D(\lambda)+\1_{\ker\mu(\{\lambda\})}$ is an automorphism of $\ker\mu(\{\lambda\})$. These results ensure the invertibility of $F+\1$ in a punctured neighbourhood of $\lambda$ because the block representation of $F$ yields around $\lambda$
$$
 F(z)+\1 =
 \begin{pmatrix}
 	\frac{1}{z-\lambda}(A(\lambda)+O(z-\lambda)) & B(\lambda)+O(z-\lambda)
 	\\
	C(\lambda)+O(z-\lambda) & D(\lambda)+\1_{\ker\mu(\{\lambda\})}+O(z-\lambda)
 \end{pmatrix},
$$
so that $\det(F(z)+\1)\ne0$ for $z\ne\lambda$ close enough to $\lambda$ because
\begin{equation} \label{eq:det}
\begin{aligned}
 \det(F(z)+\1)
 & = (z-\lambda)^{-\rank\mu(\{\lambda\})} \det
 \begin{pmatrix}
 	A(\lambda)+O(z-\lambda) & B(\lambda)+O(z-\lambda)
 	\\
	O(z-\lambda) & D(\lambda)+\1_{\ker\mu(\{\lambda\})}+O(z-\lambda)
 \end{pmatrix}
 \\
 & = (z-\lambda)^{-\rank\mu(\{\lambda\})}
 \left[
 \det A(\lambda)\det(D(\lambda)+\1_{\ker\mu(\{\lambda\})})+O(z-\lambda)
 \right].
\end{aligned}
\end{equation}
Thus, \eqref{eq:f-F} defines an analytic function $f$ in a punctured neighbourhood of $\lambda$. This function has indeed a removable singularity at $\lambda$, as follows from \eqref{eq:det}, which shows that $1/\det(F+\1)$ has a zero of order $\rank\mu(\{\lambda\})$ at $\lambda$, while any cofactor of $F+\1$ has a pole of order at most $\rank\mu(\{\lambda\})$ at $\lambda$. Thus, rewriting the relation between $f$ and $F$ in \eqref{eq:f-F} as
\begin{equation} \label{eq:fbis}
 f(z) = z^{-1}
 \left[
 \1 - 2(F(z)+\1)^{-1}
 \right],
\end{equation}
we conclude that $f$ has an analytic extension to a neighbourhood of $\lambda$.

We have seen that $F$ and $f$ are analytic around the gaps of $\supp\mu$ and $(\supp\mu)'$ respectively. Since the absolutely continuous part of $\mu$ vanishes in the gaps of $\supp\mu$, we conclude that $\re F(e^{i\theta})=0$ for every $e^{i\theta}\in\Tw\setminus\supp\mu$, which in view of \eqref{eq:f-F} is equivalent to the unitarity of $f(e^{i\theta})$. Hence, $f$ is unitary in the gaps of $\supp\mu$ and, by analyticity, also in the gaps of $(\supp\mu)'$.

Let us see now the converse, i.e. that the analyticity and unitarity of $f$ in a closed arc $\Gamma\subset\Tw$ guarantee that $\Gamma\subset\Tw\setminus(\supp\mu)'$. It suffices to show that, in $\Gamma$, the measure $\mu$ has no absolutely continuous part $\re F(e^{i\theta})\,\frac{d\theta}{2\pi}$, while the singular part is constituted by at most a finite number of mass points. The statement about the absolutely continuous part is a consequence of the unitarity of $f$ in $\Gamma$. As for the singular part of $\mu$, it is concentrated on the points $e^{i\theta}$ such that $\lim_{r\uparrow1}\tr\re F(re^{i\theta})=\infty$. The analyticity of $f$ in $\Gamma$ implies that $\det(\1-zf(z))$ has a finite number of zeros in a neighbourhood of $\Gamma$. As a consequence, there is a neighbourhood of $\Gamma$ where $F(z)=(\1+zf(z))(\1-zf(z))^{-1}$ is meromorphic with a finite number of poles. Then, the singular part of $\mu$ in $\Gamma$ is concentrated on the finitely many poles of $F$ in such an arc, which only may lead to a finite number of mass points in $\Gamma$.

\noindent{\it(ii)}
We know that the isolated mass points of $\mu$ are the poles of $F$ in the gaps of $(\supp\mu)'$, which, as follows from \eqref{eq:det}, are characterized as the poles of $\det(F+\1)$ in these gaps. Expressing \eqref{eq:fbis} as
\begin{equation} \label{eq:fbisbis}
 \frac{1}{2}(\1-zf(z))(F(z)+\1) = \1 = \frac{1}{2}(F(z)+\1)(\1-zf(z)),
\end{equation}
we see that the poles of $\det(F+\1)$ are the zeros of $\det(\1-zf(z))$, which proves the first statement of {\it(ii)}.
Therefore, we have the following expansions around an isolated mass point $\lambda$ of $\mu$,
$$
\begin{aligned}
 & \1-zf(z) = X_0 + X_1(z-\lambda) + O(z-\lambda)^2,
 & \qquad & X_0 = \1 - \lambda f(\lambda),
 \\
 & \frac{1}{2}(F(z)+\1) = \frac{Y_{-1}}{z-\lambda} + Y_0 + O(z-\lambda),
 & & Y_{-1} = -\lambda\mu(\{\lambda\}).
\end{aligned}
$$
Inserting them into \eqref{eq:fbisbis} yields, among other equations, $X_0Y_{-1} = 0$ and $Y_0X_0 + Y_{-1}X_1 = \1$. The first of these equations means that $\range Y_{-1}  \subset \ker X_0$. The second one leads to $Y_{-1}X_1\ker X_0 = \ker X_0$, which implies that $\ker X_0 \subset \range Y_{-1}$. We conclude that $\ker X_0 = \range Y_{-1}$, which is the last statement of {\it(ii)}.
\end{proof}

Schur functions may be defined in the abstract setting of operators on Hilbert spaces, as the contractive operator valued analytic functions on $\Dk$. They are related to operator valued measures via Carath\'eodory functions as in \eqref{eq:f-F}. When the underlying Hilbert space is finite-dimensional, the representation of operator valued Schur functions with respect to an orthonormal basis identifies them as matrix Schur functions.

Operator valued Schur functions are naturally connected to unitary operators on Hilbert spaces. A unitary operator $W$ on a Hilbert space $\HH$ defines for each subspace $\HH_C\subset\HH$ an operator valued measure on $\Tw$, called the spectral measure of $\HH_C$. If $W=\int t\,dE(t)$ is the spectral decomposition of $W$, the spectral measure of $\HH_C$ is the projection $P_CEP_C$ of $E$ on operators on $\HH_C$, where $P_C$ is the orthogonal projection of $\HH$ onto $\HH_C$. This associates to $\HH_C$ the Schur function $f$ of its spectral measure --which we call the {\bf Schur function of the subspace} $\HH_C$--, which admits the representation
\begin{equation} \label{eq:f-W}
 f(z) = P_C(\1-zW^* P_C^\bot)^{-1}W^* P_C = P_C(W-zP_C^\bot)^{-1}P_C,
 \qquad
 P_C^\bot=\1-P_C,
\end{equation}
and whose values are considered as operators on $\HH_C$. Indeed, every operator valued Schur function may be represented in this way for some unitary operator $W$ and subspace $\HH_C$ \cite{GV}. The representation \eqref{eq:f-W} of Schur functions will be key for our purposes. It was uncovered in \cite{GVWW,BGVW}, where the development of the quantum version of P\'olya's renewal theory for random walks \cite{Polya,Feller} identified the Schur function $z f(\overline{z})^*$ as the generating function of first returns to the subspace $\HH_C$ for the quantum walk driven by $W$.

The Schur function of the whole space $\HH$ is $W^*$. At the other end, the Schur function of a vector $\phi\in\HH$ must be understood as the scalar Schur function of the one-dimensional subspace $\spn\{\phi\}$. A remarkable instance of this arises when considering unitary operators given by {\bf CMV matrices} \cite{CMV,Simon,Watkins}, which provide the canonical form of the unitaries on Hilbert spaces and, not surprisingly, are behind widely used quantum walk models (see Sect.~\ref{sec:ssrep}). A CMV matrix has the general form
\begin{equation} \label{eq:CMV0}
 \text{\small
 $\begin{pmatrix}
 	\overline{\alpha}_0 & \rho_0\overline{\alpha}_1 & \rho_0\rho_1 & 0 & 0 &
	0 & \dots
	\\
	\rho_0 & -\alpha_0\overline{\alpha}_1 & -\alpha_0\rho_1 & 0 & 0 & 0 & \dots
	\\
	0 & \rho_1\overline{\alpha}_2 & -\alpha_1\overline{\alpha}_2 &
	\rho_2\overline{\alpha}_3 & \rho_2\rho_3 & 0 & \dots
    \\
    0 & \rho_1\rho_2 & -\alpha_1\rho_2 & -\alpha_2\overline{\alpha}_3 &
    -\alpha_2\rho_3 & 0 & \dots
    \\
    0 & 0 & 0 & \rho_3\overline{\alpha}_4 & -\alpha_3\overline{\alpha}_4 &
    \rho_4\overline{\alpha}_5 & \dots
    \\
    0 & 0 & 0 & \rho_3\rho_4 & -\alpha_3\rho_4 & -\alpha_4\overline{\alpha}_5 &
    \dots
    \\
    \dots & \dots & \dots & \dots & \dots & \dots & \dots
 \end{pmatrix}$},
 \qquad
 \begin{aligned}
 & \alpha_n\in\Dk,
 \\
 & \rho_n=\sqrt{1-|\alpha_n|^2},
 \end{aligned}
\end{equation}
or its transpose. In both cases, it is a matrix representation of the unitary operator $h(z) \mapsto zh(z)$ on $L^2_\mu$ for some probability measure $\mu$ on $\Tw$. It turns out that $\mu$ is the spectral measure of the first canonical vector $(1,0,0,0,\dots)$, representing $h(z)=1$, whose Schur function $f$ is related to the measure $\mu$ by \eqref{eq:f-F}. A key result, known as {\bf Geronimus' theorem}, states that the coefficients $\alpha_n$ defining the CMV matrix are the Schur parameters of $f$.
The CMV matrix \eqref{eq:CMV0} factorizes as
\begin{equation} \label{eq:CMV}
 \left[
 \Theta(\alpha_0) \oplus \Theta(\alpha_2) \oplus \Theta(\alpha_4) \oplus \cdots
 \right]
 \left[
 1 \oplus \Theta(\alpha_1) \oplus \Theta(\alpha_3) \oplus \cdots
 \right],
 \qquad
 \Theta(\alpha) =
 \begin{pmatrix}
 	\overline\alpha & \sqrt{1-|\alpha|^2}
	\\
	\sqrt{1-|\alpha|^2} & -\alpha
 \end{pmatrix},
\end{equation}
while the transposed one is given by the same factors but in reverse order. The doubly infinite version of CMV matrices may be written as
\begin{equation} \label{eq:eCMV}
 \left[ \bigoplus_{n\in\Ir} \Theta(\alpha_{2k-1}) \right]
 \left[ \bigoplus_{n\in\Ir} \Theta(\alpha_{2k}) \right],
\end{equation}
where the block $\Theta(\alpha_n)$ acts on $\spn\{e_n,e_{n+1}\}$ if $e_n=(\delta_{j,n})_{j\in\Ir}$ stands for the canonical basis of $\ell^2(\Ir)$.

The above results will help to the classification of topological phases in the examples of Sect.~\ref{sec:ssrep} and \ref{sec:ex} because the Schur functions of finite-dimensional subspaces $\HH_C$ with respect to unitary operators $W$ are particularly useful for the study of symmetry protected topological phases in quantum walks due to these reasons:

\begin{itemize}

\item[(A)]
Every such a Schur function $f$ has an analytic continuation through the essential gaps of $W$ which is unitary on these gaps. Hence, $f$ is unitary at the protected points $\pm1$ of an admissible unitary $W$.

\item[(B)]
For an admissible unitary $W$, the finite-dimensional unitaries $f(\pm1)$ inherit the symmetries of $W$ whenever the subspace $\HH_C$ is a sum of cells --more generally, if $\HH_C$ is symmetry invariant. If this finite sum of cells is large enough, then the representations of the symmetry type in the $\pm1$-eigenspaces of $W$ and $f(\pm1)$ are similar and, in consequence,
$$
 \six_\pm(W)=\six_\pm(f(\pm1)).
$$
The virtue of this relation is that it rewrites the symmetry indices $\six_\pm$ of infinite-dimensional unitaries as symmetry indices of finite-dimensional unitaries, for which simple explicit expressions are available.

\item[(C)]
Every perturbation $W \to VW$ by a unitary $V=V_C\oplus\1_{\HH_C^\bot}$ acting trivially on $\HH_C^\bot$ induces a perturbation $f\to fV_C^*$ on the Schur function of $\HH_C$. In particular, for $\HH_C$ a large enough sum of cells, if $W$, $VW$ are admissible walks and $VW = W_L \oplus W_R$ decouples into walks $W_{L/R}$ on a left/right half-line, then $fV_C^* = f_L \oplus f_R$ also decouples and $\six_\pm(W_{L/R})=\six_\pm(f_{L/R}(\pm1))$, so that
$$
 \sixL(W)=\six_+(f_L(1))+\six_-(f_L(-1)),
 \qquad
 \sixR(W)=\six_+(f_R(1))+\six_-(f_R(-1)).
$$
The above relations may be used to obtain the left/right indices from the decoupling of the original Schur function $f$, but also to obtain this Schur function --and thus $\six_\pm(W)$-- by combining the left/right ones $f_{L/R}$ when these later ones are easier to obtain than $f$, or with the purpose of uncovering the constraints among $\sixL(W)$, $\sixR(W)$ and $\six_-(W)$ for a specific model.

\end{itemize}

The aim of the paper is to prove the above results and to use them to classify the symmetry protected topological phases of 1D quantum walks of current interest, especially in non-translation invariant situations.

The Schur approach may be useful also in higher dimensions. In contrast to $\sixL(W)$ and $\sixR(W)$, the symmetry indices $\six_\pm(W)$ make sense for walks in any dimension, indeed for arbitrary admissible unitaries. Although $\six_\pm(W)$ are not complete for admissible walks --even in 1D--, they completely classify homotopically the admissible unitaries regardless of the dimension of the lattice. Actually, this classification is not sensible to the cell structure of the Hilbert space, which is only necessary to introduce locality conditions and, in 1D, to define left and right indices, $\sixL(W)$ and $\sixR(W)$. As we will see, the Schur representation of the symmetry indices $\six_\pm(W)$ remains valid in the general setting of admissible unitaries.

\section{Schur representation of symmetry indices \large $\six_\pm(W)$}
\label{sec:S-I}

In this section we will address items (A) and (B) stated at the end of the previous section --the discussion of (C) will wait until the next section--, i.e. we will translate to Schur functions the properties of admissible unitaries. Some of these properties are equally translated to all the related Schur functions, but others require for the associated subspace to enclose enough information about the unitary operator.
In this respect, a key role will be played by the notion of the cyclic subspace generated by a subspace $\HH_C\subset\HH$ with respect to a unitary operator $W$ on a Hilbert space $\HH$, which is the minimal subspace including the whole evolution of $\HH_C$ driven by $W$, i.e. $\overline{\sum_{n\in\Ir}W^n\HH_C} = (W^\Ir\HH_C)^{\bot\bot}$. We say that $\HH_C$ is cyclic for $W$ if its cyclic subspace is the whole Hilbert space $\HH$. The following proposition connects cyclic subspaces and eigenspaces of unitary operators, a result of interest for the discussion of topological phases.

\begin{prop} \label{prop:CYCLIC}
Let $W$ be a unitary operator on a Hilbert space $\HH$ with a cyclic subspace $\HH_C\subset\HH$, and $P_C$ the orthogonal projection of $\HH$ onto $\HH_C$. If $\lambda$ is an eigenvalue of $W$ and $E_\lambda$ is the orthogonal projection of $\HH$ onto the $\lambda$-eigenspace $\HH^{(\lambda)}$, then:
\begin{itemize}
 \item[(i)] The $\lambda$-eigenspace and the cyclic subspace are related by
 	$\HH^{(\lambda)}=E_\lambda\HH_C$ and the orthogonal decomposition
	$\HH_C=P_C\HH^{(\lambda)}\oplus(\HH_C\cap\HH^{(\lambda)\bot})$.
 \item[(ii)] The projections $P_C$ and $E_\lambda$ induce the isomorphisms
	$$
 	 \hat{P}_\lambda \colon
	 \mathop{\HH^{(\lambda)} \to P_C\HH^{(\lambda)}}
 	 \limits_{\textstyle \phi \xmapsto{\kern15pt} P_C\phi}
	 \qquad\qquad
	 \hat{E}_\lambda \colon
	 \mathop{P_C\HH^{(\lambda)} \to \HH^{(\lambda)}}
 	 \limits_{\kern9pt \textstyle \phi \xmapsto{\kern15pt} E_\lambda\phi}
	$$
 \item[(iii)] $\dim\HH^{(\lambda)} \le \dim\HH_C$.
 \end{itemize}
\end{prop}

\begin{proof}
\it(i) The equality $\HH^{(\lambda)} = E_\lambda\HH_C$ follows from $E_\lambda\HH_C = W^nE_\lambda\HH_C = E_\lambda W^n\HH_C$ and the cyclicity of $\HH_C$, while
$
 \HH_C \ominus P_C\HH^{(\lambda)} =
 \{\phi\in\HH_C: \phi\bot P_C\HH^{(\lambda)}\} =
 \{\phi\in\HH_C: \phi\bot\HH^{(\lambda)}\}
$
gives the remaining identity in {\it(i)}.

\it(ii) The linear map $\hat{P}_\lambda$ is obviously onto. To show that its kernel is trivial, assume that $W\phi=\lambda\phi$ and $P_C\phi=0$. This last condition means that $\phi\bot\HH_C$, hence $\lambda^n\phi=W^n\phi\bot W^n\HH_C$ for every $n\in\Ir$. Due to the cyclicity of $\HH_C$ we conclude that $\phi \bot \HH$, thus $\phi=0$.

We know that
$
 \HH_C \ominus P_C\HH^{(\lambda)} =
 \{\phi\in\HH_C: \phi\bot E_\lambda\HH\} =
 \{\phi\in\HH_C: E_\lambda\phi=0\},
$
which shows that $\ker\hat{E}_\lambda=\{0\}$. Since
$
 \HH^{(\lambda)}\ominus\range\hat{E}_\lambda =
 \{\phi\in\HH^{(\lambda)}: \phi \bot E_\lambda P_C\HH^{(\lambda)}\} =
 \{\phi\in\HH^{(\lambda)}: P_C\phi \bot \HH^{(\lambda)}\},
$
the orthogonal decomposition in {\it(i)} yields
$
 \HH^{(\lambda)}\ominus\range\hat{E}_\lambda =
 \{\phi\in\HH^{(\lambda)}: P_C\phi=0\} =
 \ker(\hat{P}_\lambda) = \{0\},
$
hence $\hat{E}_\lambda$ is onto.

The inequality {\it(iii)} is a direct consequence of the previous results.
\end{proof}

The theory of symmetry protected topological phases summarized in Sect.~\ref{sec:SPTP} assigns a major role to the (essential) gaps of a unitary operator, i.e. the gaps of its (essential) spectrum. In parallel, given an operator valued measure, we define its essential support as the support with the isolated mass points of finite rank mass removed, and we refer to the gaps of its (essential) support as the (essential) gaps of the measure. Then, the (essential) gaps of a unitary operator with spectral decomposition $W=\int t\,dE(t)$ coincide with the (essential) gaps of the spectral measure $E$.

Our first important result relates the essential gaps of a unitary operator to a condition for associated Schur functions. It also gives a Schur characterization of the isolated eigenvalues, identifying the  projection on a cyclic subspace of the corresponding eigenspace. This will be key to establish a Schur representation of the symmetry indices $\six_\pm$ for admissible unitaries.

\begin{thm} \label{thm:u-eg}
Let $W$ be a unitary operator on a Hilbert space $\HH$ and $P_C$ the orthogonal projection of $\HH$ onto a a finite-dimensional subspace $\HH_C\subset\HH$. If $f$ is the Schur function of $\HH_C$, then:
\begin{itemize}
 \item[(i)] $f$ has an analytic continuation through the essential gaps of $W$ 	which is unitary on such gaps.
 \item[(ii)] If $\HH_C$ is cyclic for $W$, the condition (i) characterizes the essential gaps 	of $W$, the isolated eigenvalues of $W$ are the zeros of
    $\det(\1_{\HH_C}-zf(z))$ lying on these gaps and the eigenspace
    $\HH^{(\lambda)}$ of an isolated eigenvalue $\lambda$ satisfies
	\begin{equation} \label{eq:Eschur}
	 P_C\HH^{(\lambda)}=\ker(\1_{\HH_C}-\lambda f(\lambda)).
	\end{equation}
\end{itemize}
\end{thm}

\begin{proof}
If $W=\int t\,dE(t)$ is the spectral decomposition of $W$, the essential spectrum of $W$ coincides with the essential support of $E$, which includes the essential support of the spectral measure $\mu=P_CEP_C$ associated with $\HH_C$. Besides, the essential support of $\mu$ is simply $(\supp\mu)'$ because its mass points have a finite rank mass since $\HH_C$ is finite-dimensional. Therefore, the spectral gaps of $W$ lie on $\Tw\setminus(\supp\mu)'$. Bearing in mind that $f$ is the Schur function of $\mu$, the statement {\it(i)} follows from Theorem~\ref{thm:m-eg}.{\it(i)}.

Assuming $\HH_C$ cyclic for $W$, we find that $\range\mu(\{\lambda\}) = P_CE(\{\lambda\})\HH_C = P_C\HH^{(\lambda)}$, where the last equality is due to Proposition~\ref{prop:CYCLIC}.{\it(i)}. Then, {\it(ii)} would be a consequence of Theorem~\ref{thm:m-eg} if we prove that $\mu$ and $E$ have the same isolated mass points and essential support under the cyclicity of $\HH_C$. For this it is enough to argue that these measures share strict gaps, mass points and the rank of their masses.

The mass points of $\mu$ are obviously mass points of $E$. Regarding the converse, from Proposition~\ref{prop:CYCLIC}.{\it(ii)} we know that, for any mass point $\lambda$ of $E$, $\range\mu(\{\lambda\}) = P_C\HH_\lambda$ is isomorphic to $\HH^{(\lambda)}$. Thus $\lambda$ is also a mass point of $\mu$ with
$\rank\mu(\{\lambda\}) = \dim\HH^{(\lambda)} = \rank E(\{\lambda\})$,
and \eqref{eq:Eschur} follows from \eqref{eq:muschur}.

On the other hand, the gaps of $E$ are also gaps of $\mu$. To prove the converse, consider $\triangle\subset\Tw$ such that $\mu(\triangle)=0$. This implies that the orthogonal projection $E(\triangle)$ satisfies $\|E(\triangle)W^n\phi\|^2 = \<W^n\phi|E(\triangle)W^n\phi\> = \<\phi|\mu(\triangle)\phi\> = 0$ for every $\phi\in\HH_C$ and $n\in\Ir$. Due to the cyclicity of $\HH_C$, this yields $E(\triangle)=0$. Therefore, every gap of $\mu$ is also a gap of $E$.
\end{proof}

The previous theorem shows how the information about the essential spectrum and the isolated mass points of a unitary operator is codified by Schur functions. A Schur representation of symmetry indices requires also to understand how Schur functions inherit the symmetries of a unitary operator. The answer is given by the next theorem, which uses the following terminology and notation.

\begin{defi} \label{def:symC}
Let $\rho$ be a representation of a symmetry type $\symS$ on a Hilbert space $\HH$. We say that a subspace $\HH_C\subset\HH$ is $\rho$-invariant --or simply symmetry invariant-- if it is invariant under every symmetry in $\rho$. Then, we denote by $\rho_C=\rho\upharpoonright\HH_C$ the representation of the symmetry type $\symS$ induced by $\rho$ on $\HH_C$, constituted by the symmetries $\sigma_C=\sigma\upharpoonright\HH_C$.
\end{defi}

Now we can state the result alluded to above.

\begin{thm} \label{thm:sym}
Let $W$ be an admissible unitary for a representation $\rho$ of a symmetry type $\symS$, and $\HH_C$ a finite-dimensional $\rho$-invariant subspace. If $f$ is the Schur function of $\HH_C$, then $f(\pm1)$ are finite-dimensional unitaries satisfying the symmetries in the induced representation $\rho_C=\rho\upharpoonright\HH_C$.
\end{thm}

\begin{proof}
That $f(\pm1)$ are unitary follows from the essential gap condition of admissible unitaries and Theorem~\ref{thm:u-eg}. Since $\HH_C$ is finite-dimensional, its $\rho$-invariance means that every unitary/antiunitary in $\rho$ commutes with the projection $P_C$ of $\HH$ onto $\HH_C$, and thus with $P_C^\bot=\1-P_C$. Then, any symmetry $\sigma$ of $W$ induces a similar one $\sigma_C=\sigma\upharpoonright\HH_C$ on the unitaries $f(\pm1)$, as follows by using \eqref{eq:f-W} and taking limits $r\to\pm1$ on
$$
 \sigma_C f(r) \sigma_C^* =
 \sigma_C P_C (W-rP_C^\bot)^{-1} P_C \sigma_C^* =
 P_C \sigma (W-rP_C^\bot)^{-1} \sigma^* P_C =
 P_C (\sigma W\sigma^*-rP_C^\bot)^{-1} P_C,
 \quad r\in(-1,1).
$$
\end{proof}

We are assuming that the symmetries act locally in each cell. Therefore, a finite sum of cells is an example --actually, the prime example-- of a symmetry  invariant finite-dimensional subspace. This will be a recurrent choice in the examples discussed in Sect.~\ref{sec:ssrep} and \ref{sec:ex}.

Although the induced representation on a symmetry invariant subspace belongs to the same symmetry type as the original representation, their symmetry indices may be different. However, this is not the case when the subspace is cyclic, a result which is the basis for the Schur representation of symmetry indices given by the theorem below. This is the central result of the paper.

\begin{thm} \label{thm:si}
Let $W$ be an admissible unitary for a representation $\rho$ of a symmetry type $\symS$, and $\HH_C$ a finite-dimensional $\rho$-invariant cyclic subspace. If $f$ is the Schur function of $\HH_C$, then:
\begin{itemize}
\item[(i)] $\pm1$ is an eigenvalue of $W \;\Leftrightarrow\; \pm 1$ is an eigenvalue of $f(\pm1)$.
\item[(ii)] $f(\pm1)$ are unitaries with a group of symmetries belonging to the same symmetry type $\symS$ such that
$$
 \six_\pm(W) = \six_\pm(f(\pm1)).
$$
\end{itemize}
\end{thm}

\begin{proof}
Since $\pm1$ lie in the essential gaps of an admissible unitary, Theorem~\ref{thm:u-eg}.{\it(ii)} implies that $\pm1$ is an eigenvalue of $W$ exactly when $\det(\1_{\HH_C}\mp f(\pm1))=0$, which means that $\pm1$ is an eigenvalue of $f(\pm1)$.

Theorem~\ref{thm:sym} states that $f(\pm1)$ are unitaries satisfying the symmetries in the representation $\rho_C=\rho\upharpoonright\HH_C$, which belongs to the same symmetry type $\symS$ as $\rho$. Besides, using the notation
$$
 \HH^\pm = \pm1\text{-eigenspace of } W,
 \quad
 E_\pm = \text{ orthogonal projection onto } \HH^\pm,
 \quad
 \hat{\HH}^\pm = \pm1\text{-eigenspace of } f(\pm1),
$$
we find from Theorem~\ref{thm:u-eg}.{\it(ii)} that
$\hat{\HH}^\pm=\ker(\1_{\HH_C} \mp f(\pm1))=P_C\HH^\pm$.
Therefore, Proposition~\ref{prop:CYCLIC} yields the isomorphisms
$$
 \hat{P}_\pm \colon
 \mathop{\HH^\pm \to \hat{\HH}^\pm}
 \limits_{\textstyle \phi \xmapsto{\kern10pt} P_C\phi}
 \qquad\qquad
 \hat{E}_\pm \colon
 \mathop{\hat{\HH}^\pm \to \HH^\pm}
 \limits_{\kern3pt \textstyle \phi \xmapsto{\kern8pt} E_\pm\phi}
$$
where $P_C$ is the orthogonal projection onto $\HH_C$. Let $\rho_\pm$ and $\hat\rho_\pm$ be the representations of the symmetries induced on the eigenspaces $\HH^\pm$ and $\hat\HH^\pm$ respectively. Given any symmetry $\sigma$ in $\rho$, since it commutes with $P_C$, we conclude that $\sigma_C\hat{P}_\pm\phi=\hat{P}_\pm\sigma\phi$ for $\sigma_C=\sigma\upharpoonright\HH_C$ and every $\phi\in\HH^\pm$. Thus, the isomorphism $\hat{P}_\pm$ connects these representations through
$
 \hat\rho_\pm = \hat{P}_\pm \rho_\pm \hat{P}_\pm^{-1},
$
which leads to
$
 \six_\pm(W) = \six(\rho_\pm) = \six(\hat\rho_\pm) = \six_\pm(f(\pm1)).
$
\end{proof}

The cyclicity hypothesis is a limitation of the above result because not every admissible unitary has a finite-dimensional cyclic subspace. Fortunately, this condition on the subspace can be relaxed to another one which, as we will see, imposes no constraint on admissible unitaries.

\begin{thm} \label{thm:eigC}
The results of Theorems~\ref{thm:si} hold if the cyclicity of $\HH_C$ is substituted by any of the following equivalent conditions:
\begin{itemize}
\item[(i)] The cyclic subspace generated by $\HH_C$ includes the $\pm1$-eigenspaces of $W$.
\item[(ii)] $\HH_C^\bot$ contains no $\pm1$-eigenvector of $W$.
\end{itemize}
\end{thm}

\begin{proof}
The cyclic subspace $\TT_0$ generated by $\HH_C$, as well as
$\TT_1=\TT_0^\bot$, are always $W$-invariant and inherit the $\rho$-invariance of $\HH_C$. This leads to the orthogonal decompositions $W=W_0\oplus W_1$ and $\rho=\rho_0\oplus\rho_1$, with $W_i$ a unitary on $\TT_i$ which is admissible for the representation $\rho_i$ of the symmetry type $\symS$. Condition {\it(i)} is equivalent to the statement that $W_1$ is gapped around $\pm1$, so that $\six_\pm(W)=\six_\pm(W_0)$. Besides, $W$ and $W_0$ yield the same Schur function for $\HH_C$ because this subspace lies in $\TT_0$. Then, the application of Theorem~\ref{thm:si} to $W_0$ proves that condition {\it(i)} has the same consequences as the cyclicity of $\HH_C$ concerning such a theorem.

To prove the equivalence between {\it(i)} and {\it(ii)}, we rewrite {\it(i)} by stating that $\TT_1$ has no $\pm1$-eigenvector of $W$. Then, the alluded equivalence follows from the fact that, for every eigenvector $\phi$ of $W$,
$$
 \phi\bot W^\Ir\HH_C
 \;\Leftrightarrow\;
 W^\Ir\phi\bot\HH_C
 \;\Leftrightarrow\;
 \phi\bot\HH_C.
$$
\end{proof}

Rather than the knowledge of the $\pm1$-eigenspaces of $W$, the condition given by the above proposition only needs to guarantee that some subspace --namely, $\HH_C^\bot$-- is free of $\pm1$-eigenvectors. This condition is not only less restrictive than the cyclicity of $\HH_C$ imposed in Theorem~\ref{thm:si}, but every admissible unitary has a subspace $\HH_C$ satisfying it which is constituted by a finite sum of cells, as the next proposition shows.

\begin{prop} \label{prop:exC}
Given any admissible unitary $W$, there exists a finite sum of cells $\HH_C$ which generates a cyclic subspace including the $\pm1$-eigenspaces of $W$.
\end{prop}

\begin{proof}
Let $\HH=\bigoplus_x\HH_x$ be the cell structure of the Hilbert space where $W$ acts. Consider the cyclic subspace $\TT_n$ generated by the sum of cells $\SS_n=\bigoplus_{|x|\le n}\HH_x$. Since $\TT_n$ and $\TT_n^\bot$ are both $W$-invariant, every eigenspace $\HH^{(\lambda)}$ of $W$ splits as
$\HH^{(\lambda)}=\HH^{(\lambda)}_n\oplus\KK^{(\lambda)}_n$ into the eigenspaces $\HH^{(\lambda)}_n$ and $\KK^{(\lambda)}_n$ of $W\upharpoonright\TT_n$ and $W\upharpoonright\TT_n^\bot$ respectively. Besides,
$\KK^{(\lambda)}_n\supset\KK^{(\lambda)}_{n+1}$ because $\TT_n\subset\TT_{n+1}$. If $\{0\}\ne\KK^{(\lambda)}_n=\KK^{(\lambda)}_{n+1}$ for $n\ge n_0$, then there exists a non-null vector $\phi\in\HH^{(\lambda)}$ which is orthogonal to $\SS_n$ for $n\ge n_0$, in contradiction with the fact that $\{\SS_n\}_{n\ge0}$ spans $\HH$. Therefore, either $\KK^{(\lambda)}_n=\{0\}$ for big enough $n$, or $\KK^{(\lambda)}_n\varsupsetneq\KK^{(\lambda)}_{n+1}$ for infinitely many values of $n$. The second option requires $\dim\HH^{(\lambda)}=\infty$, thus in the case $\dim\HH^{(\lambda)}<\infty$ only the first option is available, which means that $\HH^{(\lambda)}=\HH^{(\lambda)}_n\subset\TT_n$ for big enough $n$.

Since $W$ is admissible, its $\pm1$-eigenspaces $\HH^\pm$ are finite-dimensional, thus $\HH^\pm\subset\TT_n$ for big enough $n$. For any such a value of $n$, the finite sum of cells $\HH_C=\SS_n$ satisfies the requirement of the proposition.
\end{proof}

This last result guarantees that, at least theoretically, it is always possible to calculate the symmetry indices $\six_\pm(W)$ of a 1D admissible walk $W$ by applying Theorem~\ref{thm:si} to the Schur function of a finite sum of cells. The proof of Proposition~\ref{prop:exC} generalizes trivially to quantum walks in arbitrary dimension $k$, just by considering $\SS_n=\bigoplus_{\|x\|\le n}\HH_x$ for a Hilbert space with cell structure labelled by vector indices $x\in\Ir^k$. Moreover, the hypothesis of Theorem~\ref{thm:si} make no assumption on the cell structure of the Hilbert space. Therefore, the Schur approach to the symmetry indices $\six_\pm$ based on the application of Theorem~\ref{thm:si} to Schur functions of finite sums of cells also works for higher dimensional quantum walks.

For walks on the line, the results of this section can be used to obtain the indices $\sixL$ and $\sixR$ because they are given in terms of $\six_\pm$ for decoupled walks on a half-line. Nevertheless, some decoupling properties of Schur functions are useful for the calculation of left and right indices, and they are discussed in the next section.

\section{Decoupling and Schur representation of left/right indices}
\label{sec:dec}

Left and right indices are related to compact decouplings, i.e. compact perturbations $W \mapsto \widetilde{W}=W_L\oplus W_R$ which split a walk $W$ on the line into walks $W_{L/R}$ on a left/right half-line. To deal with such perturbations, we introduce the following terminology.

\begin{defi} \label{def:per}
We say that $V=\widetilde{W}W^*$ is the perturbation that connects two unitaries $W$ and $\widetilde{W}=VW$ on the same Hilbert space $\HH$. We refer to $\HH_V=(V-\1)\HH=(V^*-\1)\HH$ as the perturbation subspace. Since $V-\1$ is null on $\HH_V^\bot$ and leaves $\HH_V$ invariant, we can write $V=V_0\oplus\1_{\HH_V^\bot}$ with $V_0=V\upharpoonright\HH_V$.

A perturbation is called:
\begin{itemize}
\item Local if $\HH_V$ is included in a finite sum of cells.
\item Symmetry preserving if $VW$ shares the symmetries of $W$.
\item Decoupling if $VW=W_L\oplus W_R$ splits into unitaries $W_{L/R}$ on orthogonal subspaces $\HH_{L/R}$ (then, we refer to $\HH_V$ as the decoupling subspace). We will refer to a symmetry preserving decoupling when $\HH_L=\HH_{<a}$, $\HH_R=\HH_{\ge a}$ for some $a\in\Ir$, and the perturbation is symmetry preserving.
\end{itemize}
\end{defi}

Local perturbations are particular cases of compact ones $W \mapsto \widetilde{W}=VW$, i.e. those such that $\widetilde{W}-W$, or equivalently $V-\1$, is compact. Compact perturbations preserve both, essential gaps and essential locality. Therefore, symmetry preserving compact perturbations transform admissible walks (unitaries) into admissible walks (unitaries). In this case we know that $\sixL(VW)=\sixL(W)$ and $\sixR(VW)=\sixR(W)$. Symmetry preserving compact decouplings $VW=W_L\oplus W_R$ of an admissible walk $W$ are used for the calculation of its left and right indices since $W_{L/R}$ is then an admissible walk for the restriction of the symmetries of $W$ to $\HH_{L/R}$, and $\sixL(W)=\six(W_L)$, $\sixR(W)=\six(W_R)$.

The following result paves the way towards a Schur translation of decouplings since it translates the effect of a perturbation into the Schur function of a subspace including the perturbation subspace.

\begin{prop} \label{prop:V}
Given a unitary perturbation $\widetilde{W}=VW$ of a unitary $W$, if a subspace $\HH_C$ includes the perturbation subspace $\HH_V$, the Schur functions $f$ and $\widetilde{f}$ of $\HH_C$ with respect to $W$ and $\widetilde{W}$ are related by
$$
 \widetilde{f}=fV_C^*, \qquad V_C=V\upharpoonright\HH_C.
$$
\end{prop}

\begin{proof}
The restriction $V_C=V\upharpoonright\HH_C$ makes sense because $\HH_C$ is $V$-invariant whenever $\HH_V\subset\HH_C$. Actually, we can write $V=V_C\oplus\1_{\HH_C^\bot}$, thus $V$ commutes with the projection $P_C$ onto $\HH_C$ and $VP_C^\bot=P_C^\bot$ for $P_C^\bot=\1-P_C$. Using \eqref{eq:f-W}, this leads to
$$
 \widetilde{f}(z) = P_C(VW-zP_C^\bot)^{-1}P_C = P_C(W-zP_C^\bot)^{-1}V^*P_C
 = P_C(W-zP_C^\bot)^{-1}P_CV_C^* = f(z)V_C^*.
$$
\end{proof}

Decouplings also have Schur translations which help to study left and right indices. This calls for dealing with subspaces $\HH_C$ which respect the decomposition $\HH=\HH_L\oplus\HH_R$ behind a decoupling, i.e. such that
\begin{equation} \label{eq:HLCR}
 \HH_C = \HH_{LC} \oplus \HH_{CR}, \qquad \HH_{LC/CR}\subset\HH_{L/R}.
\end{equation}
This requirement is equivalent to any of the equivalent conditions $P_L\HH_C\subset\HH_C$ or $P_R\HH_C\subset\HH_C$, where $P_{L/R}$ is the orthogonal projection of $\HH$ onto $\HH_{L/R}$. For symmetry preserving decouplings, this is the case when $\HH_C$ is a sum of cells.

Bearing in mind the role of cyclicity in the Schur representation of symmetry indices, a couple of questions about a subspace $\HH_C$ should be answered before stating any result about decouplings in terms of Schur functions: Which perturbations leave invariant the cyclic subspace generated by $\HH_C$? Also, in view of Theorem~\ref{thm:eigC}, which perturbations leave invariant the eigenvectors orthogonal to $\HH_C$ and their eigenvalues? We will see that a sufficient condition for both is that $\HH_C$ contains the perturbation subspace.

\begin{prop} \label{prop:C-pert}
Given a unitary perturbation $VW$ of a unitary $W$, if a subspace $\HH_C$ includes the perturbation subspace $\HH_V$, then:
\begin{itemize}
\item[(i)] $W$ and $VW$ have the same eigenvectors in $\HH_C^\bot$ and they have the same eigenvalues.
\item[(ii)] $\HH_C$ generates the same cyclic subspace for $W$ and $VW$.
\end{itemize}
\end{prop}

\begin{proof}
If $\phi\in\HH_C^\bot$ then $V\phi=\phi$ because $V$ is the identity on $\HH_V^\bot\supset\HH_C^\bot$. Hence, $W\phi=\lambda\phi$ is equivalent to $VW\phi=\lambda\phi$, which proves {\it(i)}.

To prove {\it(ii)} we will show instead that the orthogonal complement to the cyclic subspace generated by $\HH_C$ is the same for $W$ and $VW$, i.e. $(W^\Ir\HH_C)^\bot=((VW)^\Ir\HH_C)^\bot$. Indeed, it is enough to see that $(W^\Ir\HH_C)^\bot\subset((VW)^\Ir\HH_C)^\bot$ because $W$ is also a perturbation of $VW$ with the same perturbation subspace. Due to the equivalence
$$
 \phi \bot W^\Ir\HH_C \;\Leftrightarrow\; W^\Ir\phi \bot \HH_C,
$$
the inclusion to prove becomes the implication
$$
 W^\Ir\phi \bot \HH_C \;\Rightarrow\; (VW)^\Ir\phi \bot \HH_C.
$$
We will prove this by showing that
$$
 W^\Ir\phi \bot \HH_C \;\Rightarrow\; (VW)^\Ir\phi = W^\Ir\phi.
$$
Suppose that $W^\Ir\phi\bot\HH_C$. Bearing in mind that $V$ and $V^{-1}$ are the identity on $\HH_C^\bot$, we find that
$$
 W^n\phi \bot \HH_C \;\Rightarrow\; VW^n\phi = W^n\phi = V^{-1}W^n\phi,
 \qquad n\in\Ir.
$$
Therefore, assuming $(VW)^n\phi=W^n\phi$ yields $(VW)^{n+1}\phi=VW^{n+1}\phi=W^{n+1}\phi$ and $(VW)^{n-1}\phi=W^{-1}V^{-1}W^n\phi=W^{n-1}\phi$. Since $V\phi=\phi$, this proves by induction that $(VW)^n\phi=W^n\phi$ for $n\in\Ir$.
\end{proof}

The previous results combine to give the following theorem, which states the main result of this section. It relates the decouplings of a walk with decouplings of Schur functions, and takes advantage of this to get a Schur representation of left and right indices.

\begin{thm} \label{thm:dec}
Let $VW=W_L\oplus W_R$ be a decoupling of a unitary $W$ and $\HH_C\supset\HH_V$ a subspace such that $\HH_{LC}:=P_L\HH_C\subset\HH_C$ (equivalently, $\HH_{CR}:=P_R\HH_C\subset\HH_C$). Then the Schur function $f$ of $\HH_C$ with respect to $W$ is related to the Schur functions $f_{L/R}$ of $\HH_{LC/CR}$ with respect to $W_{L/R}$ by
$$
 f = (f_L \oplus f_R) V_C,
 \qquad
 V_C=V\upharpoonright\HH_C.
$$
If, besides, $W$ is an admissible walk, the decoupling is symmetry preserving and $\HH_C$ is a finite-dimensional symmetry invariant subspace such that $\HH_C^\bot$ contains no $\pm1$-eigenvectors of $W$, then
$$
 \sixL(W)=\six_+(f_L(1))+\six_-(f_L(-1)),
 \qquad
 \sixR(W)=\six_+(f_R(1))+\six_-(f_R(-1)).
$$
\end{thm}

\begin{proof}
Under the initial hypothesis of the theorem, we find from Proposition~\ref{prop:V} that $fV_C^* = f_L \oplus f_R$.

If, besides, $\dim\HH_C<\infty$, then $V-\1$ is finite rank because $\HH_V\subset\HH_C$. Hence, $VW=W_L \oplus W_R$ is a compact decoupling of $W$. Since this decoupling is symmetry preserving, $W_L \oplus W_R$ is an admissible walk such that $\sixL(W)=\six(W_L)=\six_+(W_L)+\six_-(W_L)$ and $\sixR(W)=\six(W_R)=\six_+(W_R)+\six_-(W_R)$. Then, according to Theorem~\ref{thm:si}.{\it(ii)}, $\six_\pm(W_{L/R})=\six_\pm(f_{L/R}(\pm1))$ if $\HH_{L/R}\ominus\HH_{LC/CR}$ contains no $\pm1$-eigenvectors of $W_{L/R}$, which, in view of Proposition~\ref{prop:C-pert}.{\it(i)}, means that $\HH_C^\bot$ has no $\pm1$-eigenvectors of $VW$.
\end{proof}

The hypothesis in the above theorem can be considerably simplified if $\HH_C$ is a finite sum of cells, which requires a local decoupling. Then, the condition $P_L\HH_C\subset\HH_C$, as well as the symmetry invariance and finite-dimensionality of $\HH_C$, are automatic. Furthermore, Proposition~\ref{prop:exC} shows that a finite sum of cells $\HH_C$ can be always enlarged enough to guarantee that $\HH_C^\bot$ has no $\pm1$-eigenvectors of $W$. For convenience, we state separately the practical consequence of Theorem~\ref{thm:dec} derived from these remarks.

\begin{cor} \label{cor:dec}
Let $VW=W_L\oplus W_R$ be a decoupling of a unitary $W$ and $\HH_C\supset\HH_V$ a sum of cells. Then the Schur function $f$ of $\HH_C$ with respect to $W$ is related to the Schur functions $f_{L/R}$ of $\HH_{LC/CR}$ with respect to $W_{L/R}$ by
$$
 f = (f_L \oplus f_R) V_C,
 \qquad
 V_C=V\upharpoonright\HH_C.
$$
If, besides, $W$ is an admissible walk and the decoupling is local and symmetry preserving, then $\HH_C$ can be chosen as a finite sum of cells such that $\HH_C^\bot$ contains no $\pm1$-eigenvectors of $W$ and, hence,
$$
 \sixL(W)=\six_+(f_L(1))+\six_-(f_L(-1)),
 \qquad
 \sixR(W)=\six_+(f_R(1))+\six_-(f_R(-1)).
$$
\end{cor}

Eventually, an admissible walk may have a cyclic subspace constituted by a finite sum of cells. This cyclic sum of cells --enlarged to include the decoupling subspace $\HH_V$ if necessary-- is an ideal candidate to play the role of $\HH_C$ in the above corollary. This will be our choice in the examples of Sect.~\ref{sec:ssrep} and \ref{sec:ex}.

All the previous comments refer to what we could call left perturbations $W \mapsto VW$. One can also consider right perturbations $W \mapsto WV'$, and the previous results have obvious extensions to this case. The easiest way to see this is by noting that such a right perturbation is equivalent to a left perturbation $W^* \mapsto V'^*W^*$, while the transformation $W\mapsto W^*$ induces the mapping $f(z)\mapsto f(\overline{z})^*$ on Schur functions. Nevertheless, left and right perturbations have the same effect if $V'=W^*VW$, hence it seems that there is no need to consider right perturbations. However, the corresponding perturbation subspaces $\HH_V=(V-\1)\HH$ and $\HH_{V'}=(V'-\1)\HH=W^*\HH_V$ could be quite different, so that, in practice, to perform a decoupling may be simpler sometimes to use right perturbations than left ones. Actually, unless the walk is strictly local, the locality of a right perturbation is not inherited by the corresponding left one and viceversa. In the following examples we will use the freedom in the choice of left/right perturbations to make the calculations as simple as possible, resorting to the version of the above results for right perturbations when necessary.

\section{Split-step representatives of topological phases}
\label{sec:ssrep}
The experimental realization of quantum walks typically resorts to combinations of shift type operators and unitaries acting independently on each cell, known as coin operators. These constructions also provide examples of walks with non-trivial symmetry protected topological phases and, in addition, give rise to fruitful interplays with mathematical tools such as CMV matrices or Schur functions. In this section, we are going to characterize the topological phases that are realizable inside a particular model of symmetric quantum walks, the so called split-step walks, which have been studied extensively in the literature \cite{Kita}.

They describe the time-discrete dynamics of a single particle with a two-dimensional internal degree of freedom in one spatial dimension. Accordingly, the Hilbert space is $\HH = \ell_2(\Ir)\otimes\Cx^2$ and the class of evolution operators in split-step form,
\begin{equation} \label{eq:SS}
W = S_\downarrow C_2 S_\uparrow C_1,
\end{equation}
is given by the consecutive application of two partial shift operators,
\begin{equation} \label{eq:pshifts}
 S_\uparrow = \sum_{x\in\Ir}
 |x+1\uparrow\>\<x\uparrow\!| + |x\downarrow\>\<x\downarrow\!|,
 \qquad
 S_\downarrow = \sum_{x\in\Ir}
 |x\uparrow\>\<x\uparrow\!| + |x-1\downarrow\>\<x\downarrow\!|,
\end{equation}
interspersed with two coin operators $C_i = \bigoplus_{x\in\Ir} C_{i,x}$ acting locally in each cell $\HH_x$ as an $x$-dependent rotation $C_{i,x}$, which we identify with its matrix representation with respect to $\{|x\uparrow\rangle,|x\downarrow\rangle\}$,
\begin{equation} \label{eq:coins}
 C_{i,x} = R(\theta_{i,x}),
 \qquad
 R(\theta) =
 \begin{pmatrix}
	\cos\theta & -\sin\theta \\ \sin\theta & \cos\theta
 \end{pmatrix},
 \qquad
 \theta_{i,x}\in\textstyle(-\frac{\pi}{2},\frac{\pi}{2}).
\end{equation}
This model contains coined walks, which arise when $\theta_{2,x}=0$ for every $x\in\Ir$, so that $W=SC_1$ with $S=S_\downarrow S_\uparrow$ the standard conditional shift which moves the up/down states to the right/left respectively.

Split-step walks are the most widely used models to illustrate the topological phases in quantum walks. Their phase diagram has been identified in the translation invariant case corresponding to constant angles $\theta_{1,x}$, $\theta_{2,x}$ \cite{short,long} (an interactive demonstration is provided at \cite{sse}), but few results are known in the general situation \cite{long}. The Schur approach will go beyond this, giving a complete classification of topological phases for non-translation invariant split-step walks. The following theorem summarizes our results, which will be proven in the subsequent subsections.

\begin{thm} \label{thm:SS-class}
The split-step walks \eqref{eq:SS} with essential gaps around $\pm1$ exhibit 15 symmetry protected topological phases, i.e. homotopy classes of admissible walks belonging to the symmetry type $\symS = \{\1,\ph,\rv,\ch\}$, $\ph^2=\rv^2=\1$, which are summarized in the following table.
\begin{equation} \label{eq:table-16-ti-thm}
\renewcommand{\arraystretch}{1.5}
\begin{tabular}{|c|c|c|c|c|c|}
 \hline
 \;$(\sixL(W),\sixR(W),\six_-(W))$\;	& \kern-2pt
 & \;$f_R(\pm1)=\pm1$\; 				& \;$f_R(\pm1)=\mp1$\;
 & \;$f_R(\pm1)=1$\; 					& \;$f_R(\pm1)=-1$\;
 \\ \hline
 & \kern-2pt & & & &
 \\[-15.7pt]
 \hline
 $f_L(\pm1)=\pm1$		& \kern-2pt
 						& $\begin{gathered}
 							\\[-12pt]
 							(1,1,1)
 							\\[-4pt]
							\begin{aligned}
								& \scriptstyle\theta_{1,L}<0
								& & \scriptstyle\theta_{1,R}>0
								\\[-5pt]
								& \scriptstyle\theta_{2,L}=0
								& & \scriptstyle\theta_{2,R}=0
							\end{aligned}
   						   \end{gathered}$
						& $\begin{gathered}
 							\\[-12pt]
 							(1,-1,0)
 							\\[-4pt]
							\begin{aligned}
								& \scriptstyle\theta_{1,L}<0
								& & \scriptstyle\theta_{1,R}<0
								\\[-5pt]
								& \scriptstyle\theta_{2,L}=0
								& & \scriptstyle\theta_{2,R}=0
							\end{aligned}
   						   \end{gathered}$
						& $\begin{gathered}
 							\\[-12pt]
 							(1,0,0)
 							\\[-4pt]
							\begin{aligned}
								& \scriptstyle\theta_{1,L}<0
								& & \scriptstyle\theta_{1,R}=0
								\\[-5pt]
								& \scriptstyle\theta_{2,L}=0
								& & \scriptstyle\theta_{2,R}>0
							\end{aligned}
   						   \end{gathered}$ 	
						& $\begin{gathered}
 							\\[-12pt]
 							(1,0,1)
 							\\[-4pt]
							\begin{aligned}
								& \scriptstyle\theta_{1,L}<0
								& & \scriptstyle\theta_{1,R}=0
								\\[-5pt]
								& \scriptstyle\theta_{2,L}=0
								& & \scriptstyle\theta_{2,R}<0
							\end{aligned}
   						   \end{gathered}$
 \\ \hline
 $f_L(\pm1)=\mp1$ 		& \kern-2pt	
 						& $\begin{gathered}
 							\\[-12pt]
 							(-1,1,0)
 							\\[-4pt]
							\begin{aligned}
								& \scriptstyle\theta_{1,L}>0
								& & \scriptstyle\theta_{1,R}>0
								\\[-5pt]
								& \scriptstyle\theta_{2,L}=0
								& & \scriptstyle\theta_{2,R}=0
							\end{aligned}
   						   \end{gathered}$
						& $\begin{gathered}
 							\\[-12pt]
 							(-1,-1,-1)
 							\\[-4pt]
							\begin{aligned}
								& \scriptstyle\theta_{1,L}>0
								& & \scriptstyle\theta_{1,R}<0
								\\[-5pt]
								& \scriptstyle\theta_{2,L}=0
								& & \scriptstyle\theta_{2,R}=0
							\end{aligned}
   						   \end{gathered}$
						& $\begin{gathered}
 							\\[-12pt]
 							(-1,0,-1)
 							\\[-4pt]
							\begin{aligned}
								& \scriptstyle\theta_{1,L}>0
								& & \scriptstyle\theta_{1,R}=0
								\\[-5pt]
								& \scriptstyle\theta_{2,L}=0
								& & \scriptstyle\theta_{2,R}>0
							\end{aligned}
   						   \end{gathered}$
						& $\begin{gathered}
 							\\[-12pt]
 							(-1,0,0)
 							\\[-4pt]
							\begin{aligned}
								& \scriptstyle\theta_{1,L}>0
								& & \scriptstyle\theta_{1,R}=0
								\\[-5pt]
								& \scriptstyle\theta_{2,L}=0
								& & \scriptstyle\theta_{2,R}<0
							\end{aligned}
   						   \end{gathered}$   	 	
 \\ \hline
 $f_L(\pm1)=1$ 			& \kern-2pt
 						& $\begin{gathered}
 							\\[-12pt]
 							(0,1,0)
 							\\[-4pt]
							\begin{aligned}
								& \scriptstyle\theta_{1,L}=0
								& & \scriptstyle\theta_{1,R}>0
								\\[-5pt]
								& \scriptstyle\theta_{2,L}<0
								& & \scriptstyle\theta_{2,R}=0
							\end{aligned}
   						   \end{gathered}$
						& $\begin{gathered}
 							\\[-12pt]
 							(0,-1,-1)
 							\\[-4pt]
							\begin{aligned}
								& \scriptstyle\theta_{1,L}=0
								& & \scriptstyle\theta_{1,R}<0
								\\[-5pt]
								& \scriptstyle\theta_{2,L}<0
								& & \scriptstyle\theta_{2,R}=0
							\end{aligned}
   						   \end{gathered}$
						& $\begin{gathered}
 							\\[-12pt]
 							(0,0,-1)
 							\\[-4pt]
							\begin{aligned}
								& \scriptstyle\theta_{1,L}=0
								& & \scriptstyle\theta_{1,R}=0
								\\[-5pt]
								& \scriptstyle\theta_{2,L}<0
								& & \scriptstyle\theta_{2,R}>0
							\end{aligned}
   						   \end{gathered}$
						& $\begin{gathered}
 							\\[-12pt]
 							(0,0,0)
 							\\[-4pt]
							\begin{aligned}
								& \scriptstyle\theta_{1,L}=0
								& & \scriptstyle\theta_{1,R}=0
								\\[-5pt]
								& \scriptstyle\theta_{2,L}<0
								& & \scriptstyle\theta_{2,R}<0
							\end{aligned}
   						   \end{gathered}$
 \\ \hline
 $f_L(\pm1)=-1$ 		& \kern-2pt
 						& $\begin{gathered}
 							\\[-12pt]
 							(0,1,1)
 							\\[-4pt]
							\begin{aligned}
								& \scriptstyle\theta_{1,L}=0
								& & \scriptstyle\theta_{1,R}>0
								\\[-5pt]
								& \scriptstyle\theta_{2,L}>0
								& & \scriptstyle\theta_{2,R}=0
							\end{aligned}
   						   \end{gathered}$
						& $\begin{gathered}
 							\\[-12pt]
 							(0,-1,0)
 							\\[-4pt]
							\begin{aligned}
								& \scriptstyle\theta_{1,L}=0
								& & \scriptstyle\theta_{1,R}<0
								\\[-5pt]
								& \scriptstyle\theta_{2,L}>0
								& & \scriptstyle\theta_{2,R}=0
							\end{aligned}
   						   \end{gathered}$
						& $\begin{gathered}
 							\\[-12pt]
 							(0,0,0)
 							\\[-4pt]
							\begin{aligned}
								& \scriptstyle\theta_{1,L}=0
								& & \scriptstyle\theta_{1,R}=0
								\\[-5pt]
								& \scriptstyle\theta_{2,L}>0
								& & \scriptstyle\theta_{2,R}>0
							\end{aligned}
   						   \end{gathered}$
						& $\begin{gathered}
 							\\[-12pt]
 							(0,0,1)
 							\\[-4pt]
							\begin{aligned}
								& \scriptstyle\theta_{1,L}=0
								& & \scriptstyle\theta_{1,R}=0
								\\[-5pt]
								& \scriptstyle\theta_{2,L}>0
								& & \scriptstyle\theta_{2,R}<0
							\end{aligned}
   						   \end{gathered}$
 \\ \hline
\end{tabular}
\end{equation}
Here, $f_L=f_L^{x\uparrow}$ and $f_R=f_R^{x\downarrow}$ are Schur functions corresponding to the left or right part of a decoupling of the walk respectively, which are defined in \eqref{eq:fLRx}. The $x$-independent values $f_L(\pm1),f_R(\pm1)\in\{1,-1\}$ determine the topological invariants characterizing the different phases via \eqref{eq:SS-SI} and \eqref{eq:SS-SILR}. As indicated in the above table, any of these phases has a split-step representative given by a crossover of two translation invariant split-step walks with constant angles $\theta_{i,x}=\theta_{i,L/R}$, which is shown to be admissible and, in particular, gapped in Lemma~\ref{lem:SS-cross}.
\end{thm}

Before starting with the proof of the theorem, let us first give some perspective on the result. First, according to table \eqref{eq:table-equivrel}, if we identify, not only homotopic admissible walks, but also admissible walks related by compact perturbations, the above table also gives information about the corresponding equivalence classes, which are characterized by the indices $\sixL(W)$, $\sixR(W)$. Omitting the index $\six_-(W)$ in the above table we find that, under this weaker relation, the set of essentially gapped split-step walks splits into 9 classes.

On the other hand, keeping the homotopy equivalence relation but forgetting the essential locality condition enlarges the set of unitaries to consider along the homotopy deformations, leading again to a possible reduction of the number of classes. As table \eqref{eq:table-equivrel} shows, the homotopy classes of the set of admissible unitaries are labelled by the indices $\six_\pm(W)$, which, for essentially gapped split-step walks, are separately given in the table below. We find again that only 9 of such homotopy classes are present for split-step walks, although not all of them coincide with those induced by the identification of homotopic admissible walks and compact perturbations. The reduction of classes makes evident that some split-step walks sharing one of these 9 classes cannot be connected by admissible walks, thus the corresponding homotopy necessarily violates essential locality and, hence, escapes from the split-step model.
\renewcommand{\arraystretch}{1.5}
\begin{center}
\begin{tabular}{|c|c|c|c|c|c|}
 \hline
 \;$(\six_+(W),\six_-(W))$\; 	& \kern-2pt
 & \;$f_R(\pm1)=\pm1$\; 		& \;$f_R(\pm1)=\mp1$\;
 & \;$f_R(\pm1)=1$\; 			& \;$f_R(\pm1)=-1$\;
 \\ \hline
 & \kern-2pt & & & &
 \\[-15.7pt]
 \hline
 $f_L(\pm1)=\pm1$ 	& \kern-2pt	
 					& $(1,1)$ & $(0,0)$ & $(1,0)$ & $(0,1)$
 \\ \hline
 $f_L(\pm1)=\mp1$ 	& \kern-2pt
 					& $(0,0)$ & $(-1,-1)$ & $(0,-1)$ & $(-1,0)$
 \\ \hline
 $f_L(\pm1)=1$ 		& \kern-2pt
 					& $(1,0)$ & $(0,-1)$ & $(1,-1)$ & $(0,0)$
 \\ \hline
 $f_L(\pm1)=-1$ 	& \kern-2pt
 					& $(0,1)$ & $(-1,0)$ & $(0,0)$ & $(-1,1)$
 \\ \hline
\end{tabular}
\end{center}

Besides, if we enforce translation invariance, we know that $\six_\pm(W)=0$ because the essential gaps become strict gaps, hence $\sixL(W)+\sixR(W)=\six(W)=0$. The only cases in table \eqref{eq:table-16-ti-thm} which are compatible with these conditions are those corresponding to the indices $(-1,1,0)$, $(1,-1,0)$ and $(0,0,0)$. Furthermore, all these phases have translation invariant split-step representatives given in table \eqref{eq:table-16-ti-thm}, namely
$$
 (1,-1,0) \to
 \begin{cases}
 	\theta_{1,L}=\theta_{1,R}<0,
	\\
	\theta_{2,L}=\theta_{2,R}=0,
 \end{cases}
 \qquad
 (-1,1,0) \to
 \begin{cases}
 	\theta_{1,L}=\theta_{1,R}>0,
	\\
	\theta_{2,L}=\theta_{2,R}=0,
 \end{cases}
 \qquad
 (0,0,0) \to
 \begin{cases}
 	\theta_{1,L}=\theta_{1,R}=0,
	\\
	\theta_{2,L}=\theta_{2,R}\ne0.
 \end{cases}
$$
Hence, translation invariant split-step walks exhibit only 3 phases.  Since in the translation invariant case the single index $\sixR(W)$ labels the phases, we can refer to these 3 phases in short as $(1,-1,0)\equiv-1$, $(-1,1,0)\equiv1$ and $(0,0,0)\equiv0$. This result was already known for the case of constant angles $\theta_{i,x}=\theta_i$ \cite{short,long}. Note however that the previous arguments go further: They prove that the same classification holds for general translation invariant split-step walks, i.e. for any periodic sequences $\theta_{1,x}$, $\theta_{2,x}$ with an arbitrary common period.

Regarding the phase $(0,0,0)$, there should be a homotopy of admissible walks connecting the related split-step representatives with different signs for $\theta_{2,L}=\theta_{2,R}$, although this homotopy cannot be the trivial deformation of this angle because the gaps close when $\theta_{2,L}=\theta_{2,R}=0$. An example of such an homotopy --which escapes from the split-step model-- is explicitly given in \cite{ti}.

That every phase has a representative which is a crossover of translation invariant ones is a general result in the theory of symmetry protected topological phases for 1D walks \cite{long}. Nevertheless, the fact that for the split-step phases these representatives may be chosen as split-step walks, although maybe not surprising, is non-trivial. In other words, prior to our previous analysis, there was no indication that the crossovers of translation invariant split-step walks should exhaust all the split-step phases.

Furthermore, the fact that every phase may be realized by a crossover of translation invariant ones does not mean that the translation invariant situation is the end of the story. First, a rigorous mathematical treatment of the bulk-edge correspondence needs to include such crossovers in the theory,  breaking translation invariance.

Second, one cannot naively guess the phases for a non-translation invariant model by just combining those of the translation invariant case. For instance, the existence of 3 phases for translation invariant split-step walks would suggest that the 9 crossovers among them yield the same amount of phases for non-translation invariant split-step walks, a faulty line of reasoning which does not predict the 15 phases that actually exist. The extra 6 phases arise because in a crossover between two translation invariant split-step walks, changing the left or right one by another in the same phase may change the phase of the crossover. A direct inspection of table~\eqref{eq:table-16-ti-thm} reveals that this is the case when comparing the three phases $(0,0,0) \leftrightarrow (0,0,\pm1)$, but also the pairs of phases $(1,0,0) \leftrightarrow (1,0,1)$, $(-1,0,0) \leftrightarrow (-1,0,-1)$, $(0,1,0) \leftrightarrow (0,1,1)$ and $(0,-1,0) \leftrightarrow (0,-1,-1)$. These comparisons account for the extra 6 phases.

Third, the non-translation invariant general theory has physical implications concerning the bulk-edge principle which cannot be predicted in the standard translation invariant framework. To illustrate this, let us take a result from \cite{long} which provides highly non-translation invariant instances of strictly gapped split-step walks belonging to the phases $(1,-1,0)$ and $(-1,1,0)$: the correspondence
$$
 -1 \equiv (1,-1,0) \to
 \begin{cases}
 	\theta_{1,x}\in[-\frac{\pi}{2}-\epsilon,-\frac{\pi}{2}+\epsilon],
	\\
	\theta_{2,x}\in[-\epsilon',\epsilon'],
 \end{cases}
 \qquad
 1 \equiv (-1,1,0) \to
 \begin{cases}
 	\theta_{1,x}\in[\frac{\pi}{2}-\epsilon,\frac{\pi}{2}+\epsilon],
	\\
	\theta_{2,x}\in[-\epsilon',\epsilon'],
 \end{cases}
$$
holds only subject to the restriction
$$
 \sin\frac{\epsilon}{2} + \sin\frac{\epsilon'}{2}
 < \frac{1}{\sqrt{2}}.
$$
The bulk-edge correspondence described at the end of Sect.~\ref{sec:SPTP} implies that any crossover $W$ between a walk $W_1$ from the phase $(1,-1,0)$ and a walk $W_2$ from the phase $(-1,1,0)$ has indices $\sixL(W)=\sixL(W_1)=1$,   $\sixR(W)=\sixR(W_2)=1$ and $\six(W)=\sixR(W_2)-\sixR(W_1)=2$, so that the $\pm1$-eigenspaces $\HH^\pm$ of $W$ satisfy $\dim\HH^+ +\dim\HH^-\ge2$. An instance of this kind of crossover is any split-step walk $W$ such that, for large $|x|$,
\begin{equation} \label{eq:SS-cross2}
 \begin{cases}
 	\theta_{1,x}\in[-\frac{\pi}{2}-\epsilon_1,-\frac{\pi}{2}+\epsilon_1],
	& \theta_{2,x}\in[-\epsilon'_1,\epsilon'_1],
	\quad x<0,
	\\
	\theta_{1,x}\in[\frac{\pi}{2}-\epsilon_2,\frac{\pi}{2}+\epsilon_2],
	& \theta_{2,x}\in[-\epsilon'_2,\epsilon'_2],
	\quad x>0,	
 \end{cases}
 \qquad\quad
 \sin\frac{\epsilon_i}{2} + \sin\frac{\epsilon'_i}{2} < \frac{1}{\sqrt{2}}.
\end{equation}
The fact that, for such a crossover, the dimension of the combined $\pm1$-eigenspaces is bounded below by 2, is already a result beyond the scope of the translation invariant setting. In addition, the classification \eqref{eq:table-16-ti-thm} of topological phases for non-translation invariant split-step walks permits to go even further. Any split-step crossover $W$ satisfying \eqref{eq:SS-cross2} for large $|x|$ must belong to the phase $(1,1,1)$, the only split-step phase consistent with $\sixL(W)=\sixR(W)=1$. Therefore, $\six_\pm(W)=1$, which implies that actually $\dim\HH^\pm\ge1$. Since the 2-dimensional cells are cyclic, using Proposition~\ref{prop:CYCLIC}.{\it(iii)} we conclude that $1 \le \dim\HH^\pm \le 2$. This is illustrated in Fig.~\ref{approxev}.
\begin{figure}[!t]
	\includegraphics{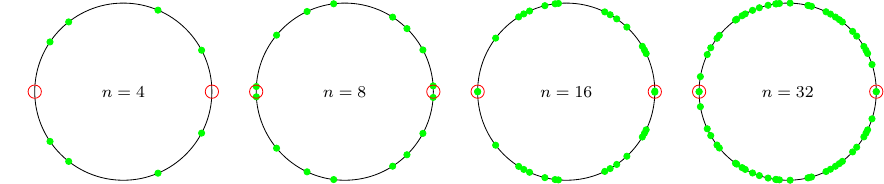}
	\includegraphics{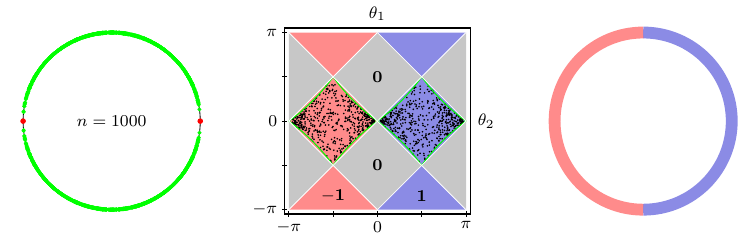}
	\caption{\label{approxev} Finite dimensional approximations of a highly non-translation invariant split-step walk \eqref{eq:SS}. The walk is defined on a circle of $n$ cells. On the right side of the circle the coin angles are picked randomly from the $+1$ phase (blue), whereas on the left side the coin angles are picked from the $-1$ phase (red) according to \eqref{eq:SS-cross2}. The bottom mid and bottom right graph visualize the randomly chosen coin angles for $n=1000$ and the geometry of the system. The green dots represent the spectrum for the given values of $n$. These approximations provide a sequence of unitaries converging in the strong sense to a split-step walk in the phase $(1,1,1)$, hence they exhibit eigenvalues approaching to $\pm1$ (in general, in complex conjugated pairs due to the symmetries).}
\end{figure}

The rest of this section is devoted to the proof of Theorem~\ref{thm:SS-class}.

\subsection{A class of symmetric quantum walks in 1D}
\label{subse:classofsymQW}
The first important question to consider is to understand which split-step walks actually correspond to admissible unitaries for some symmetry class. Instead of looking exclusively at the split-step model, we will consider a slightly generalized model, which in particular allows for higher dimensional coins, a question previously not much discussed in the literature. This will give us the chance to test the Schur machinery in highly non-trivial examples whose phase diagram was unknown so far even in the translation invariant case (see Sect.~\ref{sec:ex}). More precisely, the examples which we are going to consider will still be generated by partial shifts and coin operators, but we allow for cells of arbitrary even dimension $2d$,
\begin{equation} \label{eq:dcells}
 \HH_x=\spn\{|x\uparrow r\>,|x\downarrow r\>:r=1,2,\dots,d\}.
\end{equation}
The quantum walks are then assumed to take the form \eqref{eq:SS}, but now the partial shifts $S_{\uparrow/\downarrow}$,
\begin{equation} \label{eq:pShifts}
 S_\uparrow =
	\kern-5pt \sum_{\substack{n\in\Ir \\ r\in\{1,\dots,d\}}} \kern-5pt
	|x+1\uparrow r\>\<x\uparrow r| + |x\downarrow r\>\<x\downarrow r|, 	
 \qquad
 S_\downarrow =
	\kern-5pt \sum_{\substack{n\in\Ir \\ r\in\{1,\dots,d\}}} \kern-5pt
	|x\uparrow r\>\<x\uparrow r| + |x-1\downarrow r\>\<x\downarrow r|,
\end{equation}
move forward or backward the ``half-cells'' $\HH_x^{\uparrow/\downarrow}$ given by
\begin{equation} \label{eq:hc}
 \HH_x^\uparrow = \spn\{|x\uparrow r\>:r=1,2,\dots,d\},
 \qquad
 \HH_x^\downarrow=\spn\{|x\downarrow r\>:r=1,2,\dots,d\}.
\end{equation}
Also, the coin operators $C_i = \bigoplus_{x\in\Ir} C_{i,x}$ are given by $2d\times2d$ matrices $C_{i,x}$ representing the action of $C_i$ on $\HH_x$ with respect to the orthonormal basis $\{|x\uparrow1\rangle,\dots,|x\uparrow d\rangle,|x\downarrow1\rangle,\dots,|x\downarrow d\rangle\}$. Apart from their unitarity, we will only assume a few conditions on $C_{i,x}$ guaranteeing a chiral symmetry for $W$.

\begin{prop} \label{prop:dSS}
The operator $W = S_\downarrow C_2 S_\uparrow C_1$ given by the partial shifts \eqref{eq:pShifts} is a unitary with a chiral symmetry $\ch$ such that $\ch^2=\1$ whenever the $2d\times 2d$ matrices $C_{i,x}$ have a $d\times d$-block structure satisfying
\begin{equation} \label{eq:hC}
 C_{i,x} =
 \begin{pmatrix}
 	 A_{i,x} & \widehat{B}_{i,x} \\ B_{i,x} & A_{i,x}^*
 \end{pmatrix},
 \qquad
 \begin{aligned}
 	& \det A_{i,x}\ne0,
	& \quad & \widehat{B}_{i,x}=-A_{i,x}B_{i,x}A_{i,x}^{-1},
 	\\
	& B_{i,x}^*=B_{i,x},
	& & A_{i,x}^*A_{i,x}+B_{i,x}^2=\1_d,
 \end{aligned}
\end{equation}
where $\1_d$ stands for the $d\times d$ identity matrix. Then, $W=\tw\ch\ch$ with
$$
 \ch = \bigoplus_{x\in\Ir} \ch_x =
 \bigoplus_{x\in\Ir}
 \begin{pmatrix}
 	B_{1,x} & A_{1,x}^* \\ A_{1,x} & \widehat{B}_{1,x}
 \end{pmatrix},
 \qquad
 \tw\ch = \bigoplus_{x\in\Ir} \tw\ch_x =
 \bigoplus_{x\in\Ir}
 \begin{pmatrix}
 	B_{2,x} & A_{2,x}^* \\ A_{2,x} & \widehat{B}_{2,x}
 \end{pmatrix},
$$
where $\ch_x$ acts on $\HH_x$ while $\tw\ch_x$ acts on $\tw\HH_x=\HH_{x-1}^\downarrow\oplus\HH_x^\uparrow$.
\end{prop}

\begin{proof}
If $T$ is the spin-flip involution
\begin{equation} \label{eq:sfd}
 T = \sum_{\substack{x\in\Ir \\ r=1,\dots,d}}
 |x\downarrow r\>\<x\uparrow r|+|x\uparrow r\>\<x\downarrow r|,
\end{equation}
the operators $\ch^{(i)} = TC_i$ yield the factorization
\begin{equation} \label{eq:chs-Cd}
 W = S_b \ch^{(2)} S_f \ch^{(1)},
 \qquad
 S_b = S_\downarrow T,
 \qquad
 S_f = S_\uparrow T.
\end{equation}
From the block structure of $C_{i,x}$ we find that $\ch^{(1)}=\ch$, while $S_b\ch^{(2)}S_f=\tw\ch$. To see the later identity, let us denote $\phi_{2dx+r}=|x\uparrow r\>$ and $\phi_{2dx+d+r}=|x\downarrow r\>$. Then, $S_b$ and $S_f$ are the backward and forward shifts
$$
 S_b\phi_n=\phi_{n-d}, \qquad S_f\phi_n=\phi_{n+d},
$$
thus they are inverses of each other. We conclude that the matrix representation of $S_b\ch^{(2)}S_b^{-1}$ with respect to $\{\phi_n\}_{n\in\Ir}$ is the result of translating that of $\ch^{(2)}$ half a cell backward.

It only remains to prove that $\ch$ and $\tw\ch$ are unitary and self-adjoint. Then, $W=\tw\ch\ch$ is unitary, $\ch^2=\1$ and $\ch W \ch^* = \ch \tw\ch = W^*$, thus $\ch$ is a chiral symmetry for $W$ since $\ch$ acts locally in each cell $\HH_x$ via $\ch_x$, which is balanced because $\tr\ch_x=\tr B_{1,x}-\tr(A_{1,x}B_{1,x}A_{1,x}^{-1})=0$. Assuming $A_{i,x}$ non-singular, the rest of the conditions on the blocks of $\ch_x$ and $\tw\ch_x$ are necessary for their unitarity and self-adjointness since $\widehat{B}_{i,x}=-A_{i,x}B_{i,x}A_{i,x}^{-1}$ is equivalent to $A_{i,x}B_{i,x}+\widehat{B}_{i,x}A_{i,x}=0$. Such conditions also imply that
$$
 \widehat{B}_{i,x}^*=\widehat{B}_{i,x}
 \;\Leftrightarrow\;
 B_{i,x}A_{i,x}^*A_{i,x}=A_{i,x}^*A_{i,x}B_{i,x}
 \;\Leftrightarrow\;
 B_{i,x}(\1_d-B_{i,x}^2)=(\1_d-B_{i,x}^2)B_{i,x},
$$
proving the self-adjointness of $\ch$ and $\tw\ch$, while the remaining unitarity conditions follow by taking adjoints in $A_{i,x}B_{i,x}+\widehat{B}_{i,x}A_{i,x}=0$ and from
$$
 A_{i,x}A_{i,x}^*+\widehat{B}_{i,x}^2=\1_d
 \;\Leftrightarrow\;
 A_{i,x}A_{i,x}^*+A_{i,x}B_{i,x}^2A_{i,x}^{-1}=\1_d
 \;\Leftrightarrow\; A_{i,x}^*A_{i,x}+B_{i,x}^2=\1_d.
$$
\end{proof}

We will classify the topological phases of walks with the form given by the previous proposition, which belong to the symmetry type $\symS = \{\1,\ch\}$, $\ch^2=\1$. A particular case of these walks arises when taking arbitrary non-singular $d\times d$ matrices $A_{i,x}$ satisfying $\|A_{i,x}\|\le1$, together with
\begin{equation} \label{eq:hCs}
 B_{i,x}=\pm(\1_d-A_{i,x}^*A_{i,x})^{1/2},
 \qquad
 \widehat{B}_{i,x}=\mp(\1_d-A_{i,x}A_{i,x}^*)^{1/2}.
\end{equation}
We should point out that not all the walks given by Proposition~\ref{prop:dSS} will enter into our consideration, but only those with essential gaps around $\pm1$. Although it is not easy to translate this condition into a simple one for the blocks $A_{i,_x}$, $B_{i,x}$, this assumption will have strong consequences later on which will be central to tackle the Schur approach to the symmetry indices, leading to a complete classification of topological phases for the families of walks that we will analyze.

There is a further issue to clarify regarding this classification. The alluded walks all belong to the same symmetry type but, strictly speaking, they do not share the same representation of this symmetry type because $\ch$ depends on the doubly infinite sequences of blocks $A_{1,x}$, $B_{1,x}$. However, this is not a serious drawback since, as balanced unitaries, all the $\ch_x$ are unitarily equivalent to a diagonal one with $d$ 1s and $-1$s in the diagonal. Therefore, up to a change of basis in each cell, all the walks that we are considering correspond to the same representation of the symmetry type. Symmetry protected topological phases for these walks make about as much sense as for these unitarily equivalent pictures. Once we know this, we can work with the original matrix form of the walks which is simpler. This also holds for the Schur representation of the indices which cannot depend on the basis chosen for the matrix representation of the operator valued Schur function of a subspace.

Finally, the application of the Schur representation of the indices requires the identification of a subspace which generates a cyclic subspace containing the $\pm1$-eigenspaces. Actually, the following proposition proves that any cell is cyclic for the walks under study.

\begin{prop} \label{prop:dSS-CYCLIC}
For any $x\in\Ir$, the subspaces $\HH_x$ and $\tw\HH_x$ are cyclic for every walk $W$ with the form given in Proposition~\ref{prop:dSS}.
\end{prop}

\begin{proof}
Taking without loss the cell $\HH_0$, to prove its cyclicity it suffices to see that
$$
 \SS_n := \bigoplus_{|x|\le n}\HH_x\subset\sum_{|k|\le n} W^k\HH_0,
$$
an inclusion which is obvious for $n=0$. Assuming it for an index $n$, we will show that $\HH_{\pm n\pm1}\subset\sum_{|k|\le n+1} W^k\HH_0$, which proves the result. The induction hypothesis implies that
$$
 W\HH_{\pm n}, \; W^*\HH_{\pm n}, \; \SS_n \subset
 \sum_{|k|\le n+1} W^k\HH_0.
$$
Hence, if $P_n^\bot$ is the orthogonal projection of $\HH$ onto $\SS_n^\bot$, then
$$
 P_n^\bot W\HH_{\pm n}, P_n^\bot W^*\HH_{\pm n} \subset
 \sum_{|k|\le n+1} W^k\HH_0.
$$
From the definition of the walk and the fact that $A_{i,x}$ are non-singular, we find that
$$
 P_n^\bot W (\HH_n^\uparrow + \HH_{-n}^\downarrow) =
 \HH_{n+1}^\uparrow + \HH_{-n-1}^\downarrow,
 \qquad
 P_n^\bot W^* (\HH_n^\downarrow + \HH_{-n}^\uparrow) +
 \HH_{n+1}^\uparrow + \HH_{-n-1}^\downarrow = \HH_{n+1} + \HH_{-n-1},
$$
so that
$
 \HH_{n+1}+\HH_{-n-1} \subset
 (P_n^\bot W + P_n^\bot W^*)(\HH_n + \HH_{-n})
 \subset \sum_{|k|\le n+1} W^k\HH_0.
$

This proves that any cell $\HH_x$ is cyclic for $W$. Similar arguments hold for any subspace $\tw\HH_x$.
\end{proof}

The previous result allows us to represent the symmetry indices of any walk given in Proposition~\ref{prop:dSS} in terms of matrix Schur functions of cells. Along the previous discussion we have not distinguished between Schur functions for states --i.e. for one-dimensional subspaces-- or higher-dimensional subspaces. However, in practice, the reduction of higher-dimensional matrix Schur functions to lower-dimensional or even scalar ones will be crucial. For convenience, in the following examples we make explicit the distinction between similar matrix Schur functions of different size by using boldface notation for the Schur functions of one or more full cell Hilbert spaces.

\subsection{Split-step indices $\six_\pm(W)$}
\label{sss:SS+-}
Let us first connect the general class of quantum walks introduced in the preceding subsection back to the original split-step model defined in \eqref{eq:SS}: Setting $d=1$, Proposition~\ref{prop:dSS} yields a family of walks in a Hilbert space with 2-dimensional cells $\HH_x=\spn\{|x\uparrow\>,|x\downarrow\>\}$, given in terms of coin operators with the form
$$
 C_{i,x} = R(\theta_{i,x},\xi_{i,x}),
 \qquad
 R(\theta,\xi) =
 \begin{pmatrix}
	e^{i\xi} \cos\theta & -\sin\theta \\ \sin\theta & e^{-i\xi} \cos\theta
 \end{pmatrix},
 \qquad
 \xi_{i,x}\in[0,2\pi),
 \qquad
 \theta_{i,x}\in\textstyle(-\frac{\pi}{2},\frac{\pi}{2}).
$$
The factors $e^{\pm i\xi_{i,x}}$ are the phases of the diagonal elements of $C_{i,x}$, hence we consider only angles $\theta_{i,x}$ giving a non-negative cosine. The constraint $|\theta_{i,x}|\ne\frac{\pi}{2}$ avoiding a null diagonal is the translation of the non-singularity of the blocks $A_{i,x}$ in Proposition~\ref{prop:dSS}, which ensures the cyclicity of every cell. A change of phases $|x\uparrow\> \mapsto e^{i\nu_x} |x\uparrow\>$, $|x\downarrow\> \mapsto e^{i\omega_x} |x\downarrow\>$ transforms $R(\theta,\xi)$ into the real rotation $R(\theta):=R(\theta,0)$ if $\omega_x-\nu_x=\xi_{1,x}$ and $\nu_x-\omega_{x-1}=\xi_{2,x}$, thus we can assume without loss that $\xi_{i,x}=0$ for all $x\in\Ir$, which brings us back exactly to the split-step model from \eqref{eq:SS} as previously claimed.

As a consequence of the general discussion in the previous subsection and Proposition~\ref{prop:dSS}, split-step walks may be alternatively expressed as
\begin{equation} \label{eq:SS-CMV}
 W=\tw\ch\ch,
 \qquad
 \begin{aligned}
 	& \ch=\bigoplus_{x\in\Ir}\ch_x,
	& \quad & \ch_x = \Theta(\sin\theta_{1,x}),
 	\\
 	& \tw\ch=\bigoplus_{x\in\Ir}\tw\ch_x,
  	& & \tw\ch_x = \Theta(\sin\theta_{2,x}),
 \end{aligned}
 \qquad
 \Theta(s) =
 \begin{pmatrix}
	s & \sqrt{1-s^2}
	\\
	\sqrt{1-s^2} & -s
 \end{pmatrix},
\end{equation}
with $\ch_x$/$\tw\ch_x$ acting on $\HH_x$/$\tw\HH_x$, so that $\ch$ is a chiral symmetry of $W$ for the chosen cell structure.

When $\theta_{1,x}$ ($\theta_{2,x}$) is $\pm\frac{\pi}{2}$, the walk \eqref{eq:SS-CMV} decouples trivially because $\ch_x$ ($\tw\ch_x$) becomes a diagonal matrix, so that both involutions, $\tw\ch$ and $\ch$, and thus $W$, split at a common place. Although we will not consider such values of $\theta_{i,x}$ in our model, the related decouplings will be used later on to obtain left and right indices, as well as to express some matrix Schur functions in terms of scalar ones.

From the factorization \eqref{eq:SS-CMV} of $W$ into a couple of $2\times2$-block diagonal orthogonal matrices whose block structures do not match, $W$ is recognized as a doubly infinite CMV matrix \eqref{eq:eCMV} with real Schur parameters
\begin{equation} \label{eq:SP}
  \alpha_{2x} = \sin\theta_{2,x},
  \qquad\qquad
  \alpha_{2x+1} = \sin\theta_{1,x}.
\end{equation}
In other words, split-step walks may be identified with doubly infinite real CMV matrices, a fact that will have useful consequences.

The matrix representation of $W$ in the basis $\{|x\uparrow\>,|x\downarrow\>\}_{x\in\Ir}$ is real, thus the complex conjugation $\ph$ with respect to this basis plays the role of a particle-hole symmetry for $W$, while $\rv=\ph\ch$ is the corresponding time-reversal symmetry. Hence, split-step walks belong to the symmetry type $\symS = \{\1,\ph,\rv,\ch\}$, $\ph^2=\rv^2=\1$. Changing the cells to $\tw\HH_x=\{|x-1\downarrow\>,|x\uparrow\>\}$, then $\tw\ch$ would be a chiral symmetry for $W$, which together with $\tw\ph=\ph$ and $\tw\rv=\ph\tw\ch$ constitute another representation of the same symmetry type. We will pay attention only to the symmetries $\ph,\rv,\ch$ which act locally in each cell $\HH_x$.

We will analyze the indices $\six_\pm(W)$ of an arbitrary split-step walk $W$ by using the Schur function $\bs{f}$ of the cell $\HH_0$, which according to Proposition~\ref{prop:dSS-CYCLIC} is cyclic for $W$. To uncover the structure of the matrix Schur function $\bs{f}$ we will perform a local right decoupling $W\mapsto WV$ with decoupling subspace $\HH_V=\HH_0$, given by
$$
 V = \ch_0^* \oplus \1_{\HH_0^\bot}.
$$
$WV$ is the result of substituting $\ch_0$ by $\1_2$ in the factor $\ch$ of $W=\tw\ch\ch$. Thus, $WV=W_L\oplus W_R$ with $W_L=\tw\ch_L\ch_L$ and $W_R=\tw\ch_R\ch_R$ walks on $\HH_L=\HH_{<0}\oplus\spn\{|0\uparrow\>\}$ and $\HH_R=\spn\{|0\downarrow\>\}\oplus \HH_{>0}$ respectively, where
$$
 \tw\ch_L = \cdots \oplus \tw\ch_{-2} \oplus \tw\ch_{-1} \oplus \tw\ch_0,
 \quad
 \ch_L = \cdots \oplus \ch_{-2} \oplus \ch_{-1} \oplus 1,
 \quad
 \tw\ch_R = \tw\ch_1 \oplus \tw\ch_2 \oplus \tw\ch_3 \oplus \cdots,
 \quad
 \ch_R = 1 \oplus \ch_1 \oplus \ch_2 \oplus \cdots.
$$
The walk $W_R$ is given by a CMV matrix \eqref{eq:CMV} with Schur parameters $(\alpha_n)_{n\ge2}$, where $\alpha_n$ is as in \eqref{eq:SP}. Also, ordering the basis of $\HH_L$ as $\{|0\uparrow\>,|-1\downarrow\>,|-1\uparrow\>,\dots\}$, we may write
$$
 \tw\ch_L = \Theta(-\sin\theta_{2,0}) \oplus \Theta(-\sin\theta_{2,-1})
 \oplus \Theta(-\sin\theta_{2,-2}) \oplus \cdots,
 \quad
 \ch_L = 1 \oplus \Theta(-\sin\theta_{1,-1}) \oplus
 \Theta(-\sin\theta_{1,-2}) \oplus \cdots,
$$
identifying $W_L$ as a CMV matrix with Schur parameters $(-\alpha_{-n})_{n\ge0}$, with $\alpha_n$ given by \eqref{eq:SP}. These identifications help in the classification of topological phases of split-step walks.

Because of the broken cell $\HH_0$, the above decoupling is not symmetry preserving, the only surviving symmetry being particle-hole since the invariance under conjugation with respect to the basis $\{|x\uparrow\>,|x\downarrow\>\}_{x\in\Ir}$ does not depend on the cell structure but only on the compatibility of the decoupling with such a basis. However, this decoupling yields information about the $2\times2$ matrix Schur function $\bs{f}$ of $\HH_0$. Applying the right decoupling version of Theorem~\ref{thm:dec} to this decoupling with $\HH_V=\HH_C=\HH_0$, we get
\begin{equation} \label{eq:f-cell}
 \bs{f} = \ch_0^* (f_L \oplus f_R),
\end{equation}
with $f_L$ the Schur function of $|0\uparrow\>$ with respect to $W_L$ and $f_R$ the Schur function of $|0\downarrow\>$ with respect to $W_R$. According to \eqref{eq:SI+-fin} and Theorem~\ref{thm:si}, $\bs{f}(\pm1)$ belong to the same symmetry type $\symS$ as $W$ and
\begin{equation} \label{eq:SS-SI0}
 \six_\pm(W) = \six_\pm(\bs{f}(\pm1)) =
 \frac{1}{2} \tr \ch_0(\1_2\pm\bs{f}(\pm1)) =
 \pm\frac{1}{2} (f_L(\pm1) + f_R(\pm1)),
\end{equation}
where we have taken into account that $\ch_0$ is traceless.

Although explicit expressions for $f_{L/R}$ are not available, their properties make it possible to estimate the possible values of $f_{L/R}(\pm1)$, and thus of $\six_\pm(W)$. Theorem~\ref{thm:sym} states that the chiral symmetry $\ch$ of $W$ induces the chiral symmetry $\ch_0$ on the $2\times2$ unitaries $\bs{f}(\pm1)$. This means that $\ch_0\bs{f}(\pm1)$ are involutions, hence their eigenvalues can be only 1 or $-1$. On the other hand, \eqref{eq:f-cell} shows that $\ch_0\bs{f}(\pm1)$ are indeed diagonal with diagonal entries $f_{L/R}(\pm1)$, which thus must lie on $\{1,-1\}$. We conclude that
\begin{equation} \label{eq:SS-SI}
 \six_\pm(W) =
 \pm\frac{1}{2} (f_L(\pm1) + f_R(\pm1)) \in \{-1,0,1\}.
\end{equation}
To complete the classification of topological phases for split-step walks we should know if all these possibilities are actually present, and their connection with the possible values of left and right indices. The second question will be answered in the following subsection, while the first one will wait until Subsect.~\ref{sss:SStoph}.

\subsection{Split-step left and right indices}
\label{sss:SSLR}

The previous right decoupling is not appropriate for the calculation of left and right indices because it is not symmetry preserving. Instead, we can perform a left decoupling $W\mapsto VW$ with perturbation subspace $\HH_V=\tw\HH_0=\spn\{|-1\downarrow\>,|0\uparrow\>\}$, where
$$
 V = \tw\ch_0^* \oplus \1_{\tw\HH_0^\bot}.
$$
This amounts to substituting $\tw\ch_0$ by $\1_2$ in the left factor of $W=\tw\ch\ch$. Since this perturbation only modifies the factor $\tw\ch$ changing it by another orthogonal involution, apart from particle-hole, the decoupling preserves the chiral symmetry $\ch$ and the composition of both, i.e. time-reversal. We find that $VW=W_L\oplus W_R$ is a symmetry preserving decoupling which splits into left and right walks on $\HH_L=\HH_{<0}$ and $\HH_R=\HH_{\ge0}$. They are given again by CMV matrices $W_L=\tw\ch_L\ch_L$ and $W_R=\tw\ch_R\ch_R$, where
$$
 \tw\ch_L = \cdots \oplus \tw\ch_{-2} \oplus \tw\ch_{-1} \oplus 1,
 \quad
 \ch_L = \cdots \oplus \ch_{-3} \oplus \ch_{-2} \oplus \ch_{-1},
 \quad
 \tw\ch_R = 1 \oplus \tw\ch_1 \oplus \tw\ch_1 \oplus \cdots,
 \quad
 \ch_R = \ch_0 \oplus \ch_1 \oplus \ch_2 \oplus \cdots.
$$

From Proposition~\ref{prop:dSS-CYCLIC} we know that $\tw\HH_0$ is cyclic for $W$. Hence, as a consequence of Proposition~\ref{prop:C-pert}.{\it(ii)}, $|-1\downarrow\>$ and $|0\uparrow\>$ are cyclic vectors for $W_L$ and $W_R$ respectively. However, the corresponding Schur functions are not useful for the calculation of left and right indices because the subspaces spanned by these vectors are not symmetry invariant. We can take the cells containing the above vectors, i.e. $\HH_{-1}$ and $\HH_0$, as symmetry invariant subspaces which are cyclic for $W_L$ and $W_R$ respectively. The corresponding Schur functions, $\bs{f}_L$ and $\bs{f}_R$, will provide the left and right indices of $W$. To uncover the structure of these $2\times2$ matrix Schur functions, we can perform right perturbations $W_LV_L=W'_L\oplus1$ and $W_RV_R=1\oplus W'_R$ to decouple the states $|-1\downarrow\>$ and $|0\uparrow\>$ respectively, analogously to the right decoupling performed on $W$ in the previous subsection. This leads to the following expressions for the alluded Schur functions,
\begin{equation} \label{eq:SS-fLR}
 \bs{f}_{\!L} = \ch_{-1}^* (f_L \oplus 1),
 \qquad
 \bs{f}_{\!R} = \ch_0^* (1 \oplus f_R),
\end{equation}
with $f_{L/R}$ the Schur function of $|-1\uparrow\>$/$|0\downarrow\>$ with respect to $W'_{L/R}$.

Since the decoupling is symmetry invariant, $\ch_{-1}$/$\ch_0$ is a chiral symmetry for $\bs{f}_{L/R}(\pm1)$. Therefore, as in the previous subsection, the only possible values of $f_{L/R}(\pm1)$ are 1 and $-1$. From \eqref{eq:SI+-fin} and Corollary~\ref{cor:dec} we know that the left and right indices can be calculated as
\begin{equation} \label{eq:SS-SILR}
\begin{aligned}
 & \sixL(W) = \six_+(\bs{f}_{\!L}(1)) + \six_-(\bs{f}_{\!L}(-1)) =
 \frac{1}{2} \tr\ch_{-1}(\bs{f}_{\!L}(1) - \bs{f}_{\!L}(-1)) =
 \frac{1}{2} (f_L(1)-f_L(-1))
 \in \{-1,0,1\},
 \\
 & \sixR(W) = \six_+(\bs{f}_{\!R}(1)) + \six_-(\bs{f}_{\!R}(-1)) =
 \frac{1}{2} \tr\ch_0(\bs{f}_{\!R}(1) - \bs{f}_{\!R}(-1)) =
 \frac{1}{2} (f_R(1)-f_R(-1))
 \in \{-1,0,1\}.
\end{aligned}
\end{equation}

\bigskip

This result allows us to reobtain the expression \eqref{eq:SS-SI} for the symmetry indices $\six_\pm(W)$ with no additional effort. Consider the previous decoupling $VW=W_L\oplus W_R$ with $V = \tw\ch_0^* \oplus \1_{\tw\HH_0^\bot}$ and perturbation subspace $\HH_V=\tw\HH_0$. Applying Corollary~\ref{cor:dec} and using \eqref{eq:SS-fLR} we find that the $4\times4$ matrix Schur function $\bs{f}$ of $\HH_C=\HH_{-1}\oplus\HH_0$ with respect to $W$ is given by
$$
 \bs{f} =
 (\bs{f}_{\!L} \oplus \bs{f}_{\!R})
 (1 \oplus \tw\ch_0^* \oplus 1) =
 (\ch_{-1}^* \oplus \ch_0^*)
 (f_L \oplus \tw\ch_0^* \oplus f_R).
$$
The sum of cells $\HH_1\oplus\HH_0$ is obviously cyclic for $W$, hence from \eqref{eq:SI+-fin} and Theorem~\ref{thm:si} we obtain
\begin{equation} \label{eq:SS-SI2}
 \six_\pm(W) = \six_\pm(\bs{f}(\pm1)) =
 \frac{1}{2} \tr(\ch_{-1}\oplus\ch_0)(\1_4\pm\bs{f}(\pm1)) =
 \pm\frac{1}{2} (f_L(\pm1)+f_R(\pm1)).
\end{equation}

The above identity may appear to be the same as \eqref{eq:SS-SI}. There is, however, a slight difference between \eqref{eq:SS-SI2} and \eqref{eq:SS-SI}. To make clear this difference, let us consider in general the right decoupling $WV_x=W_{L,x}\oplus W_{R,x}$ at $\HH_{V_x}=\HH_x$ given by $V_x=\ch_x^*\oplus\1_{\HH_x^\bot}$, so that $W_{L,x}$ and $W_{R,x}$ are walks on $\HH_{L,x}=\HH_{<x}\oplus\spn\{|x\uparrow\>\}$ and $\HH_{R,x}=\spn\{|x\downarrow\>\}\oplus\HH_{>x}$ respectively. Then, we introduce the notation
\begin{equation} \label{eq:fLRx}
 f_L^{x\uparrow} =
 \begin{gathered}
 	\text{ Schur function of } |x\uparrow\>
	\\[-2pt]
 	\text{ with respect to } W_{L,x}
 \end{gathered}
 \qquad\quad
 f_R^{x\downarrow} = \begin{gathered}
 	\text{ Schur function of } |x\downarrow\>
	\\[-2pt]
 	\text{ with respect to } W_{R,x}
 \end{gathered}
\end{equation}
With this notation, \eqref{eq:SS-SI} reads as
$$
 \six_\pm(W) = \pm\frac{1}{2}
 (f_L^{0\uparrow}(\pm1) + f_R^{0\downarrow}(\pm1)),
$$
while \eqref{eq:SS-SI2} becomes
$$
 \six_\pm(W) = \pm\frac{1}{2}
 (f_L^{-1\uparrow}(\pm1) + f_R^{0\downarrow}(\pm1)).
$$
Furthermore, the relation \eqref{eq:SS-SI}, obtained using the referred right decoupling at $\HH_V=\HH_0$, may be generalized to the decoupling $W\mapsto WV_x$ at an arbitrary cell $\HH_{V_x}=\HH_x$ because all of them are cyclic for $W$. This yields,
$$
 \six_\pm(W) = \pm\frac{1}{2}
 (f_L^{x\uparrow}(\pm1) + f_R^{x\downarrow}(\pm1)),
 \qquad x\in\Ir.
$$
On the other hand, \eqref{eq:SS-SI2} follows from a left decoupling at $\HH_V=\tw\HH_0=\spn\{|-1\downarrow\>,|0\uparrow\>\}$. Its generalization to a similar decoupling with decoupling subspace $\tw\HH_x=\spn\{|x-1\downarrow\>,|x\uparrow\>\}$ leads to
$$
 \six_\pm(W) = \pm\frac{1}{2}
 (f_L^{x-1\uparrow}(\pm1) + f_R^{x\downarrow}(\pm1)),
 \qquad x\in\Ir.
$$
These two generalized relations imply that the quantities
$f_L^{x\uparrow}(\pm1)$, $f_R^{x\downarrow}(\pm1)$ are independent of the site $x$, so that \eqref{eq:SS-SI2} and \eqref{eq:SS-SI} can be considered as the same identity.

\subsection{Split-step topological phases and proof of Theorem~\ref{thm:SS-class}}
\label{sss:SStoph}

According to \eqref{eq:SS-SI} and \eqref{eq:SS-SILR}, the $x$-independent quantities
$$
 f_L(\pm1):=f_L^{x\uparrow}(\pm1),
 \;
 f_R(\pm1):=f_R^{x\downarrow}(\pm1)
 \in\{-1,1\},
$$
account for all the possibilities of the three indices $\sixL(W)$, $\sixR(W)$, $\six_-(W)$ in the case of split-step walks, which can be summarized in the single formula
$$
 (\sixL(W),\sixR(W),\six_-(W)) =
 \frac{1}{2} (f_L(1)-f_L(-1),f_R(1)-f_R(-1),-(f_L(-1)+f_R(-1))).
$$
All such possibilities are outlined in table \eqref{eq:table-16-ti-thm}, concluding that in the set of split-step walks with essential gaps around $\pm1$ there are representatives of at most 15 different topological phases, i.e. homotopy classes of admissible walks. This does not mean necessarily that two such split-step walks with the same indices are connected by a continuous curve of essentially gapped split-step walks, but by a continuous curve of admissible walks. Also, there is no guarantee yet that all the above 15 possibilities really occur among the set of essentially gapped split-step walks. We will answer affirmatively this question by finding explicit representatives for each of the possibilities listed in table \eqref{eq:table-16-ti-thm}. As indicated there, these representatives will be particular cases of the split-step walks provided by the following lemma.

\begin{lem} \label{lem:SS-cross}
The split-step walks given by
\begin{equation} \label{eq:SS-cross}
 \theta_{1,x} =
 \begin{cases}
 	\theta_{1,R}, & x>0,
 	\\
 	\theta_{1,L}, & x\le0,
 \end{cases}
 \qquad\qquad
 \theta_{2,x} =
 \begin{cases}
 	\theta_{2,R}, & x>0,
 	\\
 	\theta_{2,L}, & x\le0,
 \end{cases}
 \qquad\qquad
 \theta_{i,L/R}\in\textstyle(-\frac{\pi}{2},\frac{\pi}{2}),
\end{equation}
are essentially gapped around $\pm1$ iff $\theta_{2,L/R}\pm\theta_{1,L/R}\ne0$, thus they are admissible iff $|\theta_{2,L/R}|\ne|\theta_{1,L/R}|$. In this case, if $f_L^{x\uparrow}$ and $f_R^{x\downarrow}$ are the Schur functions defined in \eqref{eq:fLRx},
$$
 f_L(\pm1) = f_L^{x\uparrow}(\pm1) = -\sgn(\theta_{2,L} \pm \theta_{1,L}),
 \qquad
 f_R(\pm1) = f_R^{x\downarrow}(\pm1) = \sgn(\theta_{2,R} \pm \theta_{1,R}).
$$
\end{lem}

\begin{proof}
Using the notation behind \eqref{eq:fLRx}, it was shown in Subsect.~\ref{sss:SS+-} that $W_L=W_{L,0}$ and $W_R=W_{R,0}$ are CMV matrices with Schur parameters $(-\alpha_{-n})_{n\ge0}$ and $(\alpha_n)_{n\ge2}$ respectively, with $\alpha_n$ as in \eqref{eq:SP}. In our case,
$$
 (-\alpha_{-n})_{n\ge0} = (-\tilde{s}_L,-s_L,-\tilde{s}_L,-s_L,\dots),
 \qquad
 (\alpha_n)_{n\ge2} = (\tilde{s}_R,s_R,\tilde{s}_R,s_R,\dots),
 \qquad
 \begin{aligned}
 	s_{L/R}=\sin\theta_{1,L/R},
 	\\[-3pt]
	\tilde{s}_{L/R}=\sin\theta_{2,L/R}.
 \end{aligned}
$$
These are respectively the Schur parameters of $f_L=f_L^{0\uparrow}$ and $f_R=f_R^{0\downarrow}$, the Schur functions of the first vector in the basis of $\HH_L=\HH_{L,0}$ and $\HH_R=\HH_{R,0}$, as follows from the results about CMV matrices and Schur functions pointed out in Sect.~\ref{sec:SPTP}.

According to \eqref{eq:f-cell}, the Schur function of the cell $\HH_0$ is $\bs{f}=\ch_0^*(f_L \oplus f_R)$. In view of Theorem~\ref{thm:u-eg} and the cyclicity of $\HH_0$, the essential gaps of $W$ are characterized by the analyticity and unitarity of $\bs{f}$, i.e. of $f_{L/R}$. Both, $f_L$ and $f_R$ are Schur functions with 2-periodic sequences of Schur parameters. A Schur function $f$ with a 2-periodic sequence of Schur parameters $(\tilde{s},s,\tilde{s},s,\dots)$ coincides with its second Schur iterate $f_2=T_sT_{\tilde{s}}f$ coming from the Schur algorithm \eqref{eq:SA}. This leads to the equation $T_{\tilde{s}}f = T_s^{-1}f$, i.e.
$$
 \frac{1}{z} \frac{f(z)-\tilde{s}}{1-\tilde{s}f(z)} =
 \frac{zf(z)+s}{1+szf(z)},
$$
which yields
\begin{equation} \label{eq:f1}
 f(z) = \frac{z^2-1+\sqrt{\Delta}}{2z(s+\tilde{s}z)},
 \qquad
 \Delta = (1-z^2)^2+4z(s+\tilde{s}z)(sz+\tilde{s}),
\end{equation}
the branch of the square root being determined by $\sqrt{\Delta}\xrightarrow{z\to0}1$ due to the analyticity of $f$ on $\Dk$. Denoting
$$
 \theta=\arcsin s,
 \qquad
 \tilde\theta=\arcsin\tilde{s},
$$
the proof ends by showing that, except for $\tilde\theta=\mp\theta$, $f$ is analytic and unitary on $\Tw$ in a neighbourhood of $\pm1$, and in this case $f(\pm1)=\sgn(\tilde\theta\pm\theta)$.

Concerning the value of $f$ at $\pm1$, the requirement for the branch of $\sqrt{\Delta}$ implies that $f$ must be evaluated on $(-1,1)$ using the positive value of $\sqrt{\Delta}$ because $\Delta>0$ on $(-1,1)$. Therefore, \eqref{eq:f1} gives
\begin{equation} \label{eq:f1at}
 f(\pm1) = \frac{|\tilde{s} \pm s|}{\tilde{s} \pm s} =
 \sgn(\tilde\theta \pm \theta),
 \qquad
 \tilde\theta\pm\theta\ne0.
\end{equation}
As for the analyticity of $f$ on $\Tw$, the apparent singularity at $-s/\tilde{s}$ lies on $\Tw$ only when $\tilde{s}=\pm s$, but then it is removable. Thus the only singularities of $f$ on $\Tw$ come from the branch points of $\sqrt{\Delta}$, i.e. the zeros of $\Delta$. The expression of $f$ on $\Tw$,
\begin{equation} \label{eq:f1circle}
 f(e^{i\omega}) = \frac{i\sin\omega + \sqrt{\delta}}{s+\tilde{s}e^{i\omega}},
 \qquad
 \delta =
 (\cos\omega + \cos(\tilde\theta-\theta))
 (\cos\omega - \cos(\tilde\theta+\theta)),
\end{equation}
shows that the branch points $e^{i\omega}\in\Tw$ are characterized by $\cos\omega=\cos(\tilde\theta+\theta)$ or $\cos\omega=-\cos(\tilde\theta-\theta)$. None of them is $\pm1$ as long as $\tilde\theta\ne\mp\theta$, hence this condition is equivalent to the analyticity of $f$ around $\pm1$. The branch points also determine the arcs of $\Tw$ where $f$ is unitary. Since $\delta=|s+\tilde{s}e^{i\omega}|^2-(\sin\omega)^2$, we find that $|f(e^{i\omega})|=1$ iff $\delta\ge0$. On the other hand, the expression of $\delta$ in \eqref{eq:f1circle} shows that, when $\tilde\theta\ne\pm\theta$, there are four branch points which split $\Tw$ into four arcs where $\delta$ has constant sign, which is positive in the left and right arcs and negative in the upper and lower arcs. The right arc degenerates into the point 1 when $\tilde\theta=-\theta$, and the left one into the point $-1$ when $\tilde\theta=\theta$. We conclude that the arc containing $\pm1$, where $f$ is analytic and unitary, only closes for $\tilde\theta=\mp\theta$.
\end{proof}

The split-step representatives of the potential split-step phases in table \eqref{eq:table-16-ti-thm} of Theorem~\ref{thm:SS-class} may now be obtained from the walks of the previous lemma having Schur functions $f_{L/R}$ with definite parity. This completes the classification of topological phases for non-translation invariant split-step walks and we are hence ready to complete the proof of Theorem~\ref{thm:SS-class}.

\begin{proof}[Proof of Theorem~\ref{thm:SS-class}]
It only remains to prove the last statement. We will use the notation of Lemma~\ref{lem:SS-cross} and its proof. From \eqref{eq:rSA} we know that $f_{L/R}$ is an odd or even function whenever $\theta_{2,L/R}=0$ or $\theta_{1,L/R}=0$ respectively. According to Lemma~\ref{lem:SS-cross}, the following four situations yield admissible split-step crossovers:
\renewcommand{\arraystretch}{1.5}
\begin{equation} \label{eq:table-4}
\begin{tabular}{|c|c|c|c|}
 \hline
 & \kern-2pt
 & $\begin{gathered}
 		\\[-12pt]
    	\theta_{1,R}\ne0
		\\[-2pt]
		\theta_{2,R}=0
		\\[2pt]
    \end{gathered}$ 		
 & $\begin{gathered}
 		\\[-12pt]
    	\theta_{1,R}=0
		\\[-2pt]
		\theta_{2,R}\ne0
		\\[2pt]
    \end{gathered}$
 \\ \hline
 & \kern-2pt & &
 \\[-15.7pt]
 \hline
  \kern9pt
  $\begin{gathered}
 		\\[-12pt]
    	\theta_{1,L}\ne0
		\\[-2pt]
		\theta_{2,L}=0
		\\[2pt]
   \end{gathered}$
   \kern5pt				& \kern-2pt	
 						& $\begin{aligned}
    						& \scriptstyle\theta_{1,L}\ne0
							& \kern-3pt & \scriptstyle\theta_{1,R}\ne0
							\\[-3pt]
							& \scriptstyle\theta_{2,L}=0
							& \kern-3pt & \scriptstyle\theta_{2,R}=0
    					   \end{aligned}$
					    & $\begin{aligned}
    						& \scriptstyle\theta_{1,L}\ne0
							& \kern-3pt & \scriptstyle\theta_{1,R}=0
							\\[-3pt]
							& \scriptstyle\theta_{2,L}=0
							& \kern-3pt & \scriptstyle\theta_{2,R}\ne0
    					   \end{aligned}$
 \\ \hline
 $\begin{gathered}
 		\\[-12pt]
    	\theta_{1,L}=0
		\\[-2pt]
		\theta_{2,L}\ne0
		\\[2pt]
   \end{gathered}$ 		& \kern-2pt
 						& $\begin{aligned}
    						& \scriptstyle\theta_{1,L}=0
							& \kern-3pt & \scriptstyle\theta_{1,R}\ne0
							\\[-3pt]
							& \scriptstyle\theta_{2,L}\ne0
							& \kern-3pt & \scriptstyle\theta_{2,R}=0
    					   \end{aligned}$
						& $\begin{aligned}
    						& \scriptstyle\theta_{1,L}=0
							& \kern-3pt & \scriptstyle\theta_{1,R}=0
							\\[-3pt]
							& \scriptstyle\theta_{2,L}\ne0
							& \kern-3pt & \scriptstyle\theta_{2,R}\ne0
    					   \end{aligned}$
 \\ \hline
\end{tabular}
\end{equation}
Consider the upper-left option of \eqref{eq:table-4}. Then, $f_L$ and $f_R$ are odd Schur functions with Schur parameters $(0,-s_L,0,-s_L,\dots)$ and $(0,s_R,0,s_R,\dots)$ respectively. Lemma~\ref{lem:SS-cross} also provides the explicit values $f_L(\pm1)=\mp\sgn(\theta_{1,L})$ and $f_R(\pm1)=\pm\sgn(\theta_{1,R})$. In consequence, playing with the signs of $\theta_{1,L}$ and $\theta_{1,R}$, this crossover gives representatives for 4 of the possibilities in table \eqref{eq:table-16-ti-thm}, all those satisfying $f_{L/R}(-1)=-f_{L/R}(1)$, which correspond to the 4 options in the upper-left corner of \eqref{eq:table-16-ti-thm}.

In the lower-right case of \eqref{eq:table-4} $f_L$ and $f_R$ are even Schur functions with Schur parameters $(-\tilde{s}_L,0,-\tilde{s}_L,0,\dots)$ and $(\tilde{s}_R,0,\tilde{s}_R,0,\dots)$ respectively. This implies that $f_{L/R}(-1)=f_{L/R}(1)$, while Lemma~\ref{lem:SS-cross} states that $f_L(\pm1)=-\sgn(\theta_{2,L})$ and $f_R(\pm1)=\sgn(\theta_{2,R})$. Different choices for the signs of $\theta_{2,L}$ and $\theta_{2,R}$ lead to representatives of the 4 possibilities in the lower-right corner of table \eqref{eq:table-16-ti-thm}.

A similar analysis shows that the upper-right option of \eqref{eq:table-4} yields odd/even Schur functions $f_{L/R}$, giving representatives of the 4 cases in the upper-right corner of \eqref{eq:table-16-ti-thm}, while the 4 possibilities in the lower-left corner of \eqref{eq:table-16-ti-thm} have representatives with even/odd Schur functions $f_{L/R}$ arising from the lower-left case of \eqref{eq:table-4}.
\end{proof}

Finally, let us comment on the changes in the split-step phases when reducing the number of symmetries that must be preserved by the homotopies. Forgetting some symmetries of the split-step walk gives in general a coarser classification of topological phases. However, a closer look at the discussion giving rise to the split-step phases \eqref{eq:table-16-ti-thm} reveals that the same classification holds even if we do not take into account the particle-hole and time-reversal symmetries. In contrast, if we just keep time-reversal or particle-hole, the phases landscape changes. The triviality of the symmetry type $\symS = \{\1,\rv\}$, $\rv^2=\1$, implies that split-step walks have a single phase when assuming only time reversal. Also, if we consider the split-step walks as instances of the symmetry type $\symS = \{\1,\ph\}$, $\ph^2=\1$, there are at most 8 phases that come from the possible combinations of the indices $\sixL(W),\sixR(W),\six_-(W)\in\Ir_2$. Actually, these 8 phases are present for split-step walks. To see this, note that the discussions of subsections~\ref{sss:SS+-} and \ref{sss:SSLR} remain unchanged in this situation, except for the formulas giving the symmetry indices of the finite-dimensional unitaries $\bs{f}(\pm1)$, which should use now the first line of \eqref{eq:SI+-fin} instead of the second one. Then, \eqref{eq:SS-SI} becomes
$$
 (-1)^{\six_\pm(W)} =
 (-1)^{\six_\pm(\bs{f}(\pm1))} =
 \det(\mp\bs{f}(\pm1)) =
 \det(\mp\ch_0^*(f_L(\pm1) \oplus f_R(\pm1))) =
 -f_L(\pm1) f_R(\pm1),
$$
while \eqref{eq:SS-SILR} translates as
$$
\begin{aligned}
 (-1)^{\sixL(W)} & =
 (-1)^{\six_+(\bs{f}_L(1))} (-1)^{\six_-(\bs{f}_L(-1))} =
 \det(-\bs{f}_L(1)) \det(\bs{f}_L(-1))
 \\ & =
 \det(-\ch_{-1}^*(f_L(1) \oplus 1))
 \det(\ch_{-1}^*(f_L(-1) \oplus 1)) =
 f_L(1) f_L(-1),
 \\
 (-1)^{\sixR(W)} & =
 (-1)^{\six_+(\bs{f}_R(1))} (-1)^{\six_-(\bs{f}_R(-1))} =
 \det(-\bs{f}_R(1)) \det(\bs{f}_R(-1))
 \\ & =
 \det(-\ch_0^*(1 \oplus f_R(1)))
 \det(\ch_0^*(1 \oplus f_R(-1))) =
 f_R(1) f_R(-1).
\end{aligned}
$$
The analogue of table~\eqref{eq:table-16-ti-thm} for this setting
$$
\renewcommand{\arraystretch}{1.5}
\begin{tabular}{|c|c|c|c|c|c|}
 \hline
 \;$(\sixL(W),\sixR(W),\six_-(W))$\;	& \kern-2pt
 & \;$f_R(\pm1)=\pm1$\; 				& \;$f_R(\pm1)=\mp1$\;
 & \;$f_R(\pm1)=1$\; 					& \;$f_R(\pm1)=-1$\;
 \\ \hline
 & \kern-2pt & & & &
 \\[-15.7pt]
 \hline
 $f_L(\pm1)=\pm1$		& \kern-2pt
 						& $(1,1,1)$	 & $(1,1,0)$   & $(1,0,0)$   & $(1,0,1)$
 \\ \hline
 $f_L(\pm1)=\mp1$ 		& \kern-2pt	
 						& $(1,1,0)$  & $(1,1,1)$   & $(1,0,1)$   & $(1,0,0)$
 \\ \hline
 $f_L(\pm1)=1$ 			& \kern-2pt
 						& $(0,1,0)$  & $(0,1,1)$   & $(0,0,1)$   & $(0,0,0)$
 \\ \hline
 $f_L(\pm1)=-1$ 		& \kern-2pt
 						& $(0,1,1)$  & $(0,1,0)$   & $(0,0,0)$   & $(0,0,1)$
 \\ \hline
\end{tabular}
$$
is the result of applying to \eqref{eq:table-16-ti-thm} the ``forget homomorphism'' $\Ir \xrightarrow{\mod\,2} \Ir_2$ \cite{long} between the index groups of the symmetry types $\symS = \{\1,\ph,\rv,\ch\}$, $\ph^2=\rv^2=\ch^2=\1$ and $\symS = \{\1,\ph\}$, $\ph^2=\1$. The above table shows that the crossovers of translation invariant split-step walks given in \eqref{eq:table-16-ti-thm} cover the alluded 8 phases, which, except for the phase $(0,0,0)$, follow by combining in pairs the 15 phases corresponding to the symmetry type $\symS = \{\1,\ph,\rv,\ch\}$, $\ph^2=\rv^2=\1$. This means that split-step walks from different components of such pairs may be continuously connected preserving essential locality, essential gaps and particle-hole, but paying the price of violating the chiral symmetry along the path.

\section{Further Examples}
\label{sec:ex}
Starting from the general model of walks with an arbitrary even number of internal degrees of freedom introduced in Subsect.~\ref{subse:classofsymQW}, in this section we will explore the phase diagram of coined quantum walks with chiral symmetry.

\subsection{Chiral coined walks with higher-dimensional coins}
\label{ss:ex-Cd}

We will consider 1D coined walks with a chiral symmetry, but with cells of arbitrary even dimension $2d$ \eqref{eq:dcells}. They are the specialization of the walks given in Proposition~\ref{prop:dSS} for a trivial coin operator $C_2=\1$, i.e.
\begin{equation} \label{eq:dC}
 W=SC_1,
 \qquad
 S = S_\downarrow S_\uparrow
   = \sum_{\substack{x\in\Ir \\ r=1,\dots,d}}
	|x+1\uparrow r\>\<x\uparrow r| +
	|x-1\downarrow r\>\<x\downarrow r|,
\end{equation}
with $S$ the conditional shift which moves forward/backward the ``half-cells'' $\HH_x^{\uparrow/\downarrow}$. When the coin $C_1$ has the form \eqref{eq:hC}, $W=\tw\ch\ch$ where
\begin{equation} \label{eq:tch-C}
 \tw\ch = \bigoplus_{x\in\Ir} \tw\ch_x =
 \bigoplus_{x\in\Ir}
 \begin{pmatrix}
 	0 & \1_d \\[-5pt] \1_d & 0
 \end{pmatrix},
\end{equation}
with $\tw\ch_x$ acting on $\tw\HH_x$, and $\ch$ being a chiral symmetry of $W$.
Thus, we refer to these walks as chiral coined walks. We will classify the topological phases of the essentially gapped chiral coined walks with $2d\times 2d$ coins.

The phase structure for $d=1$ follows from the previous results on split-step walks. For coined walks, the Schur functions $f_{L/R}$ in Theorem~\ref{thm:SS-class} are odd because their sequences of Schur parameters have null Schur parameters at odd places. This corresponds to the 4 phases $(1,1,1)$, $(1,-1,0)$, $(-1,1,0)$, $(-1,-1,-1)$ in the upper-left corner of table~\eqref{eq:table-16-ti-thm}, from which only $(1,-1,0)$ and $(-1,1,0)$ survive in the translation invariant case because then the essential gaps become strict gaps and thus $\six_\pm(W)=0$. For $d>1$, the topological phases of coined walks appear not to have been discussed elsewhere, even in the translation invariant case.

\subsubsection{Chiral coined walk indices $\six_\pm(W)$}
\label{sss:C+-}

According to Theorem~\ref{thm:si}, the cyclicity of any cell allows us to obtain the symmetry indices $\six_\pm(W)$ out of the $2d\times2d$ matrix Schur function $\bs{f}$ for the cell $\HH_0$. Using \eqref{eq:SI+-fin} we find that
$$
 \six_\pm(W) = \six_\pm(\bs{f}(\pm1)) =
 \frac{1}{2} \tr \ch_0(\1_{2d}\pm\bs{f}(\pm1)) =
 \pm\frac{1}{2} \tr \ch_0 \bs{f}(\pm1)
 \in \{-d,-d+1,\dots,-1,0,1,\dots,d-1,d\},
$$
because $\ch_0$ is traceless and $\ch_0\bs{f}(\pm1)$ is a $2d\times2d$ unitary involution since, due to Theorem~\ref{thm:sym}, $\ch_0$ is a chiral symmetry of $\bs{f}(\pm1)$. As in the split-step case, we can say something extra about the Schur function $\bs{f}$ by using the local right decoupling $WV=W_L\oplus W_R$ with decoupling subspace $\HH_0$ given by
$$
 V = \ch_0^* \oplus \1_{\HH_0^\bot}.
$$
The left part $W_L$ of the decoupled walk acts on $\HH_L=\HH_{<0}\oplus\HH_0^\uparrow$, while the right part $W_R$ is a walk on $\HH_R=\HH_0^\downarrow\oplus\HH_{>0}$. The version of Theorem~\ref{thm:dec} for right decouplings implies that
\begin{equation} \label{eq:f-Cell}
 \bs{f} = \ch_0^*(f_L \oplus f_R),
\end{equation}
where $f_{L/R}$ is the $d\times d$ matrix Schur function of $\HH_0^{\uparrow/\downarrow}$ with respect to $W_{L/R}$. Therefore,
\begin{equation} \label{eq:C-SI}
 \six_\pm(W) =
 \pm\frac{1}{2} (\tr f_L(\pm1) + \tr f_R(\pm1)).
\end{equation}
Since $f_{L/R}(\pm1)$ are $d\times d$ matrix unitary involutions in which $\ch_0\bs{f}(\pm1)$ splits,
\begin{equation} \label{eq:trfLR}
 \tr f_L(\pm1), \; \tr f_R(\pm1) \in \{-d,-d+2,\dots,d-2,d\}.
\end{equation}
We will see that these traces also yield the left and right indices of the walk.

\subsubsection{Chiral coined walk left and right indices}
\label{sss:CLR}

To discuss the left and right indices we will perform the symmetry preserving left decoupling $W\mapsto VW=W_L\oplus W_R$ with decoupling subspace $\tw\HH_0$, given by
$$
 V = \tw\ch_0^* \oplus \1_{\tw\HH_0^\bot}.
$$
The walks $W_L$ and $W_R$ act on $\HH_L=\HH_{<0}$ and $\HH_R=\HH_{\ge0}$ respectively. Proposition~\ref{prop:dSS-CYCLIC} states that $\tw\HH_0=\HH_{-1}^\downarrow\oplus\HH_0^\uparrow$ is cyclic for $W$, thus Proposition~\ref{prop:C-pert}.{\it(ii)} implies that $\HH_{-1}^\downarrow$ and $\HH_0^\uparrow$ are cyclic for $W_L$ and $W_R$ respectively. Taking $\HH_{-1}$ and $\HH_0$ as symmetry invariant subspaces which are cyclic for $W_L$ and $W_R$, the corresponding Schur functions, $\bs{f}_{\!L}$ and $\bs{f}_{\!R}$, will provide the left and right indices of $W$.

The structure of these $2d\times2d$ matrix Schur functions follows from right perturbations of $W_{L/R}$ which are analogous to those performed on $W$ in the previous subsection. The right perturbations $W_LV_L=W'_L\oplus\1_d$ and $W_RV_R=\1_d\oplus W'_R$ which decouple respectively the subspaces $\HH_{-1}^\downarrow$ and $\HH_0^\uparrow$, yield the following expressions for such Schur functions,
$$
 \bs{f}_{\!L} = \ch_{-1}^*(f_L \oplus \1_d),
 \qquad
 \bs{f}_{\!R} = \ch_0^*(\1_d \oplus f_R),
$$
where $f_{L/R}$ is the $d\times d$ matrix Schur function of $\HH_{-1/0}^{\uparrow/\downarrow}$ with respect to $W'_{L/R}$. From \eqref{eq:SI+-fin} and Corollary~\ref{cor:dec},
\begin{equation} \label{eq:C-SILR}
\begin{aligned}
 & \sixL(W) = \six_+(\bs{f}_{\!L}(1)) + \six_-(\bs{f}_{\!L}(-1)) =
 \frac{1}{2} \tr\ch_{-1}(\bs{f}_L(1) - \bs{f}_L(-1)) =
 \frac{1}{2} (\tr f_L(1) - \tr f_L(-1)),
 \\
 & \sixR(W) = \six_+(\bs{f}_{\!R}(1)) - \six_-(\bs{f}_{\!R}(-1)) =
 \frac{1}{2} \tr\ch_0(\bs{f}_R(1) - \bs{f}_R(-1)) =
 \frac{1}{2} (\tr f_R(1) - \tr f_R(-1)).
\end{aligned}
\end{equation}
The initial left decoupling preserves the chiral symmetry $\ch$, hence $W_L$ and $W_R$ have respectively the chiral symmetries $\ch_L=\oplus_{x<0}\ch_x$ and $\ch_R=\oplus_{x\ge0}\ch_x$. Therefore, $\ch_{-1/0}$ is a chiral symmetry for the unitary $\bs{f}_{\!L/R}(\pm1)$, which means that $\ch_{-1}\bs{f}_L(\pm1)$ and $\ch_0\bs{f}_R(\pm1)$ are involutions. In consequence, the $d\times d$ unitaries $f_{L/R}(\pm1)$ are also involutions, leading once again to the constraint \eqref{eq:trfLR}.

Analogously to the split-step case, we can use the above result to reobtain the expression \eqref{eq:C-SI} for the symmetry indices $\six_\pm(W)$. Applying Corollary~\ref{cor:dec} to the previous left decoupling $VW=W_L\oplus W_R$ we find that the $4d\times4d$ matrix Schur function $\bs{f}$ of $\HH_{-1}\oplus\HH_0$ with respect to $W$ is given by
$$
 \bs{f} =
 (\bs{f}_{\!L}\oplus \bs{f}_{\!R}) (\1_d \oplus \tw\ch_0^* \oplus \1_d) =
 (\ch_{-1}^* \oplus \ch_0^*)(f_L \oplus \tw\ch_0^* \oplus f_R).
$$
Since $\HH_1\oplus\HH_0$ is cyclic for $W$, Theorem~\ref{thm:si} and \eqref{eq:SI+-fin} yield
\begin{equation} \label{eq:C-SI2}
 \six_\pm(W) = \six_\pm(\bs{f}(\pm1)) =
 \frac{1}{2} \tr(\ch_{-1}\oplus\ch_0)(\1_{4d}\pm\bs{f}(\pm1)) =
 \pm\frac{1}{2} (\tr f_L(\pm1) + \tr f_R(\pm1)).
\end{equation}

Nevertheless, there is an apparent difference between \eqref{eq:C-SI2} and \eqref{eq:C-SI}. Let $WV_x=W_{L,x}\oplus W_{R,x}$ be the right decoupling at $\HH_{V_x}=\HH_x$ given by $V_x=\ch_x^*\oplus\1_{\HH_x^\bot}$, so that $W_{L,x}$ and $W_{R,x}$ are walks on $\HH_{L,x}=\HH_{<x}\oplus\HH_x^\uparrow$ and $\HH_{R,x}=\HH_x^\downarrow\oplus\HH_{>x}$ respectively. Then, using the notation
\begin{equation} \label{eq:fLRx-C}
 f_L^{x\uparrow} =
 \begin{gathered}
 	\text{ Schur function of } \HH_x^\uparrow
	\\[-2pt]
 	\text{ with respect to } W_{L,x}
 \end{gathered}
 \qquad\quad
 f_R^{x\downarrow} = \begin{gathered}
 	\text{ Schur function of } \HH_x^\downarrow
	\\[-2pt]
 	\text{ with respect to } W_{R,x}
 \end{gathered}
\end{equation}
\eqref{eq:C-SI} becomes
$$
 \six_\pm(W) = \pm\frac{1}{2}
 (\tr f_L^{0\uparrow}(\pm1) + \tr f_R^{0\downarrow}(\pm1)),
$$
while \eqref{eq:C-SI2} reads as
$$
 \six_\pm(W) = \pm\frac{1}{2}
 (\tr f_L^{-1\uparrow}(\pm1) + \tr f_R^{0\downarrow}(\pm1)).
$$
Generalizing \eqref{eq:C-SI} to the decoupling $W\mapsto WV_x$ at an arbitrary cell $\HH_{V_x}=\HH_x$ gives
$$
 \six_\pm(W) = \pm\frac{1}{2}
 (\tr f_L^{x\uparrow}(\pm1) + \tr f_R^{x\downarrow}(\pm1)),
 \qquad x\in\Ir.
$$
Also, \eqref{eq:C-SI2}, which follows from a left decoupling at $\tw\HH_0=\HH_{-1}^\downarrow\oplus\HH_0^\uparrow$, may be generalized to a similar decoupling at $\tw\HH_x=\HH_{x-1}^\downarrow\oplus\HH_x^\uparrow$, leading to
$$
 \six_\pm(W) = \pm\frac{1}{2}
 (\tr f_L^{x-1\uparrow}(\pm1) + \tr f_R^{x\downarrow}(\pm1)),
 \qquad x\in\Ir.
$$
As a consequence, $\tr f_L^{x\uparrow}(\pm1)$ and $\tr f_R^{x\downarrow}(\pm1)$ are independent of $x\in\Ir$, which shows that \eqref{eq:C-SI} and \eqref{eq:C-SI2} are actually the same identity.

\subsubsection{Chiral coined walk topological phases}
\label{sss:Ctoph}

The identities \eqref{eq:C-SI} and \eqref{eq:C-SILR} show that, for any even dimension $2d$ of the coins \eqref{eq:hC}, the topological phases of the chiral coined walks are characterized by the $x$-independent traces
\begin{equation} \label{eq:trfLR2}
 \tr f_L(\pm1):=\tr f_L^{x\uparrow}(\pm1),
 \;
 \tr f_R(\pm1):=\tr f_R^{x\downarrow}(\pm1)
 \in\{-d,-d+2,\dots,d-2,d\},
\end{equation}
which determine the three indices $\sixL(W)$, $\sixR(W)$, $\six_-(W)$. However, a special property of coined walks yields an additional constraint on $\tr f_{L/R}(\pm1)$ which reduces the number of possible phases.

\begin{prop} \label{prop:C-odd}
The Schur function of a cell with respect to a coined walk \eqref{eq:dC} is an odd function.
\end{prop}

\begin{proof}
The even and odd subspaces
$$
 \HH_e = \bigoplus_{a\in\Ir} \HH_{2a},
 \qquad
 \HH_o = \bigoplus_{a\in\Ir} \HH_{2a+1},
$$
are exchanged by a coined walk $W=SC$, i.e. $W\HH_e=\HH_o$ and $W\HH_o=\HH_e$. Besides, these subspaces are invariant for the orthogonal projection $P_x$ onto a cell $\HH_x$, and also for the complementary projection $P_x^\bot=\1-P_x$. On the other hand, \eqref{eq:f-W} yields the following power expansion for the Schur function $\bs{f}$ of $\HH_x$,
\begin{equation} \label{eq:f-exp}
 \bs{f}(z) = \sum_{n\ge0} P_x(W^*P_x^\bot)^nW^*P_x z^n,
 \qquad
 z\in\Dk.
\end{equation}
The proposition follows from the fact that $(W^*P_x^\bot)^nW^*$ exchanges $\HH_e$ and $\HH_o$ when $n$ is even.
\end{proof}

With the help of the previous result we can finally classify the topological phases of the chiral coined walks.

\begin{thm} \label{thm:C-class}
The coined walks \eqref{eq:dC} with $2d\times 2d$ coins $C_{1,x}$ satisfying \eqref{eq:hC} and essential gaps around $\pm1$ exhibit $(d+1)^2$ symmetry protected topological phases, i.e. homotopy classes of admissible walks belonging to the symmetry type $\symS = \{\1,\ch\}$, $\ch^2=\1$. If $f_L=f_L^{x\uparrow}$ and $f_R=f_R^{x\downarrow}$ are the Schur functions defined in \eqref{eq:fLRx-C}, these phases are characterized by the $x$-independent traces \eqref{eq:trfLR2} according to
$$
 \sixL(W)=\tr f_L(1),
 \;
 \sixR(W)=\tr f_R(1) \in \{-d,-d+2,\dots,d-2,d\},
 \qquad
 \six_-(W)=\frac{1}{2}(\sixL(W)+\sixR(W)).
$$
Any of these phases has a representative which is an orthogonal sum of split-step walks.
Translation invariance reduces the number of phases to $d+1$, which arise by imposing $\tr f_L(1)=-\tr f_R(1)$.
\end{thm}

\begin{proof}
Proposition~\ref{prop:C-odd} implies that the Schur functions $f_{L/R}$ in \eqref{eq:f-Cell} are odd, hence $f_{L/R}(-1)=-f_{L/R}(1)$. Bearing in mind \eqref{eq:C-SI} and \eqref{eq:C-SILR}, this yields the expressions of the symmetry indices given in the theorem, giving rise to at most $(d+1)^2$ phases.
To prove that chiral coined walks with $2d\times2d$ coins actually exhibit all these $(d+1)^2$ phases, we need to find a representative for each of the possible values of $\tr f_{L/R}(1)$. For this purpose we consider a diagonal choice for the $d\times d$ blocks of the coin $C_{1,x}$, namely
\begin{equation} \label{eq:ABdiag}
 A_{1,x} = \scriptsize
 \begin{pmatrix}
 \cos\theta_x^{(1)}
 \\
 & \kern-18pt \cos\theta_x^{(2)}
 \\[-3pt]
 & & \kern-18pt \ddots
 \\[-4pt]
 & & & \kern-12pt \cos\theta_x^{(d)}
 \end{pmatrix},
 \qquad
 B_{1,x} = \scriptsize
 \begin{pmatrix}
 \sin\theta_x^{(1)}
 \\
 & \kern-18pt \sin\theta_x^{(2)}
 \\[-3pt]
 & & \kern-18pt \ddots
 \\[-4pt]
 & & & \kern-12pt \sin\theta_x^{(d)}
 \end{pmatrix},
 \qquad
 \normalsize \theta_x^{(r)} \in \textstyle (-\frac{\pi}{2},\frac{\pi}{2}).
\end{equation}
This choice only links states $|x\uparrow r\>$, $|x\downarrow r\>$ with the same index $r$, thus $W=\oplus_{r=1}^dW^{(r)}$ decouples into coined walks $W^{(r)}$ of split-step type with coin angles $\theta_{1,x}^{(r)}=\theta_x^{(r)}$ and $\theta_{2,x}^{(r)}=0$. The essential gap around $\pm1$ for $W$ obviously translates into similar ones for $W^{(r)}$, while the chiral symmetry of $W$ is orthogonal sum of those of $W^{(r)}$. In view of the operator representation \eqref{eq:f-W}, the Schur function $\bs{f}$ of the cell $\HH_0$ is the $2d\times2d$ matrix function
$$
 \bs{f} = \scriptsize
 \begin{pmatrix}
 f^{(1)}
 \\[-2pt]
 & \kern-9pt f^{(2)}
 \\[-6pt]
 & & \kern-12pt \ddots
 \\[-5pt]
 & & & \kern-4pt f^{(d)}
 \end{pmatrix},
$$
where $f^{(r)}$ is the $2\times2$ matrix Schur function of $\spn\{|0\uparrow r\>,|0\downarrow r\>\}$ with respect to $W^{(r)}$. The previously discussed decouplings of $W$ follow from similar decouplings of the walks $W^{(r)}$, already analyzed in the split-step example. We conclude that the $d\times d$ matrix functions $f_{L/R}$ in \eqref{eq:trfLR2} are diagonal, and their diagonals are the scalar functions $f_{L/R}^{(r)}$ which play a similar role for the coined walks $W^{(r)}$ of split-step type. As for any coined walk, the functions $f_{L/R}^{(r)}$ are odd, and the split-step example shows that their possible values at 1 are $\pm1$. Therefore, we can get any of the values of $\tr f_{L/R}(1)=\sum_{r=1}^d f_{L/R}^{(r)}(1)$ in \eqref{eq:trfLR2} just by choosing properly the split-step walks $W^{(r)}$.

Under translation invariance the essential gaps become strict gaps, thus $\six_\pm(W)=0$ and $\sixL(W)=-\sixR(W)$. Hence $\tr f_L(1)=-\tr f_R(1)$, as follows from the previous results. Moreover, there is a chiral coined walk with $\tr f_R(1)=-\tr f_L(1)$ being any integer in $\{-d,-d+2,\dots,d-2,d\}$. Just choose $A_{1,x}$ and $B_{1,x}$ as in \eqref{eq:ABdiag} with site independent angles $\theta_x^{(r)}=\theta^{(r)}\in(-\pi/2,\pi/2)\setminus\{0\}$. According to Theorem~\ref{thm:SS-class}, this makes $f_R^{(r)}(1)=-f_L^{(r)}(1)=\sgn\theta^{(r)}$. Hence, we get any of the alluded values of of $\tr f_{L/R}(1)=\sum_{r=1}^d f_{L/R}^{(r)}(1)$ by a suitable choice of the number of positive and negative angles $\theta^{(r)}$.
\end{proof}

For $d=1$, the above theorem gives the 4 phases already observed for the coined walk specialization of the split-step walk (upper-left corner of table~\eqref{eq:table-16-ti-thm} in Theorem~\ref{thm:SS-class}). As for $d>1$, for instance, in the case $d=2$, chiral coined walks present 9 phases whose symmetry indices are given in the table below, whose principal anti-diagonal corresponds to the 3 phases of the translation invariant case.
$$
\renewcommand{\arraystretch}{1.5}
\begin{tabular}{|c|c|c|c|c|}
 \hline
 \;$(\sixL(W),\sixR(W),\six_-(W))$\;	& \kern-2pt
 & \;$\tr f_R(\pm1)=2$\; 			
 & \;$\tr f_R(\pm1)=0$\; 			
 & \;$\tr f_R(\pm1)=-2$\; 		
 \\ \hline
 & \kern-2pt & & &
 \\[-15.7pt]
 \hline
 $\tr f_L(\pm1)=2$	 	& \kern-2pt
 						& $(2,2,2)$	    & $(2,0,1)$		& $(2,-2,0)$
 \\ \hline
 $\tr f_L(\pm1)=0$  	& \kern-2pt	
 						& $(0,2,1)$ 	& $(0,0,0)$ 	& $(0,-2,-1)$
 \\ \hline
 $\tr f_L(\pm1)=-2$ 	& \kern-2pt
 						& $(-2,2,0)$  	& $(-2,0,-1)$   & $(-2,-2,-2)$
 \\ \hline
\end{tabular}
$$

We have found that chiral coined walks with $2d\times2d$ coins exhibit $(d+1)^2$ topological phases, among which only $d+1$ arise in the translation invariant case. The later result is valid, not only for $x$-independent coins $C_{1,x}$, but also for periodic ones with an arbitrary period. Combining the arguments given in Theorem~\ref{thm:C-class} with the results of the split-step example we can see that every phase has a representative which is a crossover of translation invariant chiral coined walks. This agrees with the fact that the total number of phases coincides with the possible combinations of translation invariant ones. However, this naive counting, already proved to be misleading in the split-step case, will change even more dramatically in the following example.

\subsection{Walks unitarily equivalent to chiral coined walks}
\label{ss:ex-ueCd}

We will consider now the specialization of the walks in Proposition~\ref{prop:dSS} for a trivial coin $C_1=\1$. They have the form
\begin{equation} \label{eq:uedC}
 W = S_\downarrow C_2 S_\uparrow = S_\uparrow^* (SC_2) S_\uparrow,
\end{equation}
where $S = S_\downarrow S_\uparrow = S_\uparrow S_\downarrow$ is the conditional shift given in \eqref{eq:dC}. We will assume that the coins $C_{2,x}$ satisfy \eqref{eq:hC}, so that $W$ has a chiral symmetry. Although these walks are unitarily equivalent to chiral coined walks, we cannot infer their topological phases from those obtained in the previous example because the equivalence is given by the unitary operator $S_\uparrow$, which does not preserve the cell structure. Indeed, the topological phases of this model turn out to be quite different from those of chiral coined walks.

According to Proposition~\ref{prop:dSS}, $W=\tw\ch\ch$ where
$$
 \ch = \bigoplus_{x\in\Ir} \ch_x =
 \bigoplus_{x\in\Ir}
 \begin{pmatrix}
 	0 & \1_d \\[-5pt] \1_d & 0
 \end{pmatrix}
$$
is the chiral symmetry of $W$. Comparing with the previous example, we see that $W^*=\ch\tw\ch$ is a chiral coined walk, but with respect to the cell structure $\HH=\oplus_{x\in\Ir}\tw\HH_x$. Therefore, analyzing the topological phases of chiral walks of the form \eqref{eq:uedC} is the same as studying the topological phases of chiral coined walks \eqref{eq:dC} with respect to the modified cells $\tw\HH_x$ and the corresponding chiral symmetry $\tw\ch$ \eqref{eq:tch-C}. This alternative view of the walks \eqref{eq:uedC} allows us to translate to them most of the results from the previous example. Actually, the only --but key-- difference comes from the parity of the the Schur function of a cell.

\begin{prop} \label{prop:ueC-even}
The Schur function of a cell with respect a walk \eqref{eq:uedC} is an even function.
\end{prop}

\begin{proof}
Without loss, we can consider the Schur function $\bs{f}$ of $\HH_0$. We need to show that the coefficients $P_0(W^*P_0^\bot)^nW^*P_0=(P_0W(P_0^\bot W)^nP_0)^*$ of its power expansion \eqref{eq:f-exp} vanish for odd $n$. This is a direct consequence of the following inclusion
\begin{equation} \label{eq:inc2}
 W(P_0^\bot W)^n\HH_0 \subset
 (\tw\HH_{-n} \oplus \cdots \oplus \tw\HH_{-3} \oplus \tw\HH_{-1})
 \oplus (\tw\HH_2 \oplus \tw\HH_4 \oplus \cdots \oplus \tw\HH_{n+1}),
 \qquad \text{odd } n,
\end{equation}
which we will prove by induction. From the definition \eqref{eq:uedC} of the walk we find that
\begin{equation} \label{eq:inc}
 W\HH_x^\uparrow \subset \tw\HH_{x+1},
 \qquad
 W\HH_x^\downarrow \subset \tw\HH_x,
 \qquad
 W\HH_x \subset \tw\HH_x \oplus \tw\HH_{x+1},
 \qquad
 W\tw\HH_x \subset \tw\HH_{x-1} \oplus \tw\HH_{x+1}.
\end{equation}
For $n=1$, \eqref{eq:inc} yields
$$
 WP_0^\bot W\HH_0 \subset WP_0^\bot(\tw\HH_0\oplus\tw\HH_1)
 = W(\HH_{-1}^\downarrow\oplus\HH_1^\uparrow)
 \subset \tw\HH_{-1}\oplus\tw\HH_2.
$$
Assuming \eqref{eq:inc2} for a given even index $n$, using again \eqref{eq:inc} we obtain
$$
\begin{aligned}
 W(P_0^\bot W)^{n+2}\HH_0
 & \subset WP_0^\bot W
 \left(
 (\tw\HH_{-n} \oplus \cdots \oplus \tw\HH_{-3} \oplus \tw\HH_{-1})
 \oplus
 (\tw\HH_2 \oplus \tw\HH_4 \oplus \cdots \oplus \tw\HH_{n+1})
 \right)
 \\
 & \subset WP_0^\bot
 \left(
 (\tw\HH_{-n-1} \oplus \cdots \oplus \tw\HH_{-2} \oplus \tw\HH_0)
 \oplus
 (\tw\HH_1 \oplus \tw\HH_3 \oplus \cdots \oplus \tw\HH_{n+2})
 \right)
 \\
 & = W
 \left(
 (\tw\HH_{-n-1} \oplus \cdots \oplus \tw\HH_{-2} \oplus \HH_{-1}^\downarrow)
 \oplus
 (\HH_1^\uparrow \oplus \tw\HH_3 \oplus \cdots \oplus \tw\HH_{n+2})
 \right)
 \\
 & =
 (\tw\HH_{-n-2} \oplus \cdots \oplus \tw\HH_{-3} \oplus \tw\HH_{-1})
 \oplus
 (\tw\HH_2 \oplus \tw\HH_4 \oplus \cdots \oplus \tw\HH_{n+3}).
\end{aligned}
$$
\end{proof}

Compared with chiral coined walks, the change in the parity of the Schur function of a cell for the walks \eqref{eq:uedC} is the origin of a very different structure of topological phases for these walks with a chiral symmetry.

\begin{thm} \label{thm:ueC-class}
The walks \eqref{eq:uedC} with $2d\times 2d$ coins $C_{2,x}$ satisfying \eqref{eq:hC} and essential gaps around $\pm1$ exhibit $2d+1$ symmetry protected topological phases, i.e. homotopy classes of admissible walks belonging to the symmetry type $\symS = \{\1,\ch\}$, $\ch^2=\1$. If $f_L=f_L^{x\uparrow}$ and $f_R=f_R^{x\downarrow}$ are the Schur functions defined in \eqref{eq:fLRx-C}, these phases are characterized by the $x$-independent traces \eqref{eq:trfLR2} according to
$$
 \sixL(W)=\sixR(W)=0,
 \qquad
 \six_-(W)=-\frac{1}{2}(\tr f_L(1)+\tr f_R(1)) \in \{-d,-d+1,\dots,d-1,d\}.
$$
Any of these phases has a representative which is an orthogonal sum of split-step walks.
Translation invariance reduces the phases to a single one with $\sixL(W)=\sixR(W)=\six_-(W)=0$.
\end{thm}

\begin{proof}
The parity of $\bs{f}$ is also inherited by the functions $f_{L/R}$ appearing in the decomposition $\bs{f}=\ch_0^*(f_L\oplus f_R)$ generated by the right decoupling of $W$ at $\HH_0$. As in the coined walk case, the conclusion is that the phases of the walks \eqref{eq:uedC} are controlled by the values of the even functions $f_{L/R}$ at $\pm1$, which are $d\times d$ unitary involutions related by $f_{L/R}(-1)=f_{L/R}(1)$. Their traces have the possible values
$$
 \tr f_L(1), \; \tr f_R(1) \in \{-d,-d+1,\dots,d-1,d\},
$$
giving rise to the symmetry indices
$$
\begin{gathered}
 \sixL(W) = \frac{1}{2} (\tr f_L(1) - \tr f_L(-1)) = 0,
 \qquad
 \sixR(W) = \frac{1}{2} (\tr f_R(1) - \tr f_R(-1)) = 0,
 \\
 \six_-(W) = -\frac{1}{2} (\tr f_L(-1) + \tr f_R(-1)) =
 -\frac{1}{2} (\tr f_L(1) + \tr f_R(1)).
\end{gathered}
$$
Since the left and right indices are null, the phases of this model are characterized by a single symmetry index, $\six_-(W)\in\{-d,-d+1,\dots,d-1,d\}$. Therefore, these walks have at most $2d+1$ phases. Just as in the case of chiral coined walks, we can build explicit examples for all the alluded $2d+1$ potential phases by taking orthogonal sums of split-step walks $W^{(r)}$, in this case with $\theta_{1,x}^{(r)}=0$. Therefore, all the alluded phases are present in this model. Among these $2d+1$ phases, only $(0,0,0)$ is compatible with the strict gaps around $\pm1$ imposed by translation invariance, and indeed any orthogonal sum of split-step walks $W^{(r)}$ with $x$-independent angles $\theta_{1,x}^{(r)}=0$ and $\theta_{2,x}^{(r)}={\theta}^{(r)}\in(-\pi/2,\pi/2)\setminus\{0\}$ is a translation invariant representative of such a phase.
\end{proof}

For $d=1$, the 3 phases given by the previous theorem are those already observed for the split-step walks with $\theta_{1,x}=0$ (lower-right corner of table~\ref{eq:table-16-ti-thm}). As an example for $d>1$, the following table relates the symmetry indices with the possible values of $\tr f_{L/R}(\pm1)$ for the 5 phases in the case $d=2$.
$$
\renewcommand{\arraystretch}{1.5}
\begin{tabular}{|c|c|c|c|c|}
 \hline
 \;$(\sixL(W),\sixR(W),\six_-(W))$\;	& \kern-2pt
 & \;$\tr f_R(\pm1)=2$\; 			
 & \;$\tr f_R(\pm1)=0$\; 			
 & \;$\tr f_R(\pm1)=-2$\; 		
 \\ \hline
 & \kern-2pt & & &
 \\[-15.7pt]
 \hline
 $\tr f_L(\pm1)=2$			& \kern-2pt
 						 	& $(0,0,-2)$ 	& $(0,0,-1)$	& $(0,0,0)$
 \\ \hline
 $\tr f_L(\pm1)=0$  		& \kern-2pt	
 						 	& $(0,0,-1)$ 	& $(0,0,0)$ 	& $(0,0,1)$
 \\ \hline
 $\tr f_L(\pm1)=-2$ 		& \kern-2pt
 						 	& $(0,0,0)$  	& $(0,0,1)$   	& $(0,0,2)$
 \\ \hline
\end{tabular}
$$
The anti-diagonals yield the same phase, the principal one corresponding to the only phase present under translation invariance.

Despite exhibiting a single phase in the translation invariant case --i.e. periodic coins $C_{2,x}$ with any period--, the number $2d+1$ of non-translation invariant phases of this model grows linearly with $d$. This example clearly shows that the number of topological phases of a model cannot be inferred from the number of its translation invariant phases.

\section*{Acknowledgements}

C.~Cedzich acknowledges support by the Excellence Initiative of the German Federal and State Governments (Grant ZUK 81) and the DFG (project B01 of CRC 183).

T.~Geib and R.~F.~Werner acknowledge support from the ERC grant DQSIM, the DFG SFB 1227 DQmat, and the European project SIQS.

The work of L. Vel\'azquez is partially supported by the research project MTM2017-89941-P from Ministerio de Econom\'{\i}a, Industria y Competitividad of Spain and the European Regional Development Fund (ERDF), by project UAL18-FQM-B025-A (UAL/CECEU/FEDER) and by Project E26\_17R of Diputaci\'on General de Arag\'on (Spain) and the ERDF 2014-2020 ``Construyendo Europa desde Arag\'on''.

A.~H.~Werner thanks the Alexander von Humboldt Foundation for its support with
a Feodor Lynen Fellowship and acknowledges financial support from the
European Research Council (ERC Grant Agreement No. 337603) and VILLUM
FONDEN via the QMATH Centre of Excellence (Grant No. 10059).

\bibliography{tphbib}

\end{document}